\documentclass[11pt]{article}
\usepackage{amsfonts}
\usepackage{amssymb}
\usepackage{amsthm}
\usepackage{amsmath}
\usepackage{bbm}
\usepackage{graphicx}
\usepackage{latexsym}
\usepackage{mathrsfs}
\usepackage[all]{xy}

\usepackage{geometry}
\newcommand{\cA}{\mathcal{A}}
\newcommand{\cB}{\mathcal{B}}
\newcommand{\cC}{\mathcal{C}}
\newcommand{\cD}{\mathcal{D}}

\newcommand{\cF}{\mathcal{F}}

\newcommand{\cH}{\mathcal{H}}
\newcommand{\cI}{\mathcal{I}}

\newcommand{\cO}{\mathcal{O}}

\newcommand{\rC}{\mathrm{C}}

\newcommand{\rF}{\mathrm{F}}

\newcommand{\rH}{\mathrm{H}}

\newcommand{\rf}{\mathrm{f}}

\newcommand{\rh}{\mathrm{h}}

\newcommand{\bC}{\mathbb{C}}

\newcommand{\bN}{\mathbb{N}}

\newcommand{\bZ}{\mathbb{Z}}

\newcommand{\si}{\sigma}

\newcommand{\e}{\epsilon}

\newcommand {\norm}[1]{\Vert{#1}\Vert}    

\newcommand{\ad}{\mathrm{ad}}

\author{\textsc{Giuseppe Ruzzi$^{1}$ and Ezio Vasselli$^{2}$}\footnote{Both the authors are supported by the  EU network ``Noncommutative Geometry" MRTN-CT-2006-0031962.}\\
  \null\\
\small{$^{1}$Dipartimento di Matematica, Universit\`a di Roma ``Tor Vergata'',}\\
\small{Via della Ricerca Scientifica, I-00133 Roma,  Italy.}  \\
\small{\texttt{ruzzi@mat.uniroma2.it}} \\[5pt]
\small{$^{2}$Dipartimento di Matematica, Universit\`a di Roma ``La Sapienza'',}\\
\small{Piazzale Aldo Moro 5, I-00185 Roma, Italy.}\\
 \small{\texttt{ vasselli@mat.uniroma2.it  }}\\[20pt]
}
        
\date{}

\title{\textsc{Representations\\of\\nets of $\rC^*$-algebras over $S^1$}} 

\begin{document}
\maketitle

\begin{abstract}
In recent times a new kind of representations has been used to describe superselection
sectors of the observable net over a curved spacetime, taking into account of the effects of the fundamental group of the spacetime.
Using this notion of representation, we prove that any net of $\rC^*$-algebras over 
$S^1$ admits faithful representations, 
and when the net is covariant under $\mathrm{Diff}(S^1)$,  it admits representations 
covariant under any amenable subgroup  of $\mathrm{Diff}(S^1)$. 
\end{abstract}

\markboth{Contents}{Contents}


  \theoremstyle{plain}
  \newtheorem{definition}{Definition}[section]
  \newtheorem{theorem}[definition]{Theorem}
  \newtheorem{proposition}[definition]{Proposition}
  \newtheorem{corollary}[definition]{Corollary}
  \newtheorem{lemma}[definition]{Lemma}

  \theoremstyle{definition}
  \newtheorem{remark}[definition]{Remark}
    \newtheorem{example}[definition]{Example}

\theoremstyle{definition}
  \newtheorem{ass}{\underline{\textit{Assumption}}}[section]


\numberwithin{equation}{section}

\section{Introduction}

Nets of $\rC^*$-algebras are the basic objects of study in algebraic quantum field theory and, 
as well-known to the specialists, encode the basic idea that any 
%
%
suitable region\footnote{Open, relatively compact and simply connected subsets of the spacetime.} $Y$ of a spacetime defines an abstract $\rC^*$-algebra $\cA_Y$,
interpreted as the one generated by the quantum observables localized in $Y$; 
from this assumption it is natural to require that there are inclusions morphisms
$\jmath_{Y'Y} : \cA_Y \to \cA_{Y'}$,
$\forall Y \subseteq Y'$,
which, for coherence, must fulfil the equalities
\[
\jmath_{Y''Y'} \circ \jmath_{Y'Y} = \jmath_{Y''Y}
\ \ , \ \
\forall Y \subseteq Y' \subseteq Y''
\]
(see \cite{HK,Haa,Ara}).
In the Minkowski spacetime $X$ the set of  regions $Y \subset X$
is upward directed under inclusion, so the pair 
$(\cA,\jmath)$, $\cA := \{ \cA_Y \}$, $\jmath := \{ \jmath_{Y'Y} \}$,
is indeed a net and we can construct the inductive limit $\vec{\cA} := \lim (\cA,\jmath)$.
In this way, families of Hilbert space representations of $\cA_Y$, $Y \subset X$,
coherent with the inclusion morphisms (that we call {\em Hilbert space representations} of the net)
are obtained by considering representations of $\vec{\cA}$.
This point is important for the applications, because crucial physical properties of the quantum system described 
by $(\cA,\jmath)$, like the charge structure of elementary particles, 
are encoded by certain Hilbert space representations of the net, called {\em sectors} (\cite{DHR1,DHR2,BF}).\\
\indent Now, general relativity and conformal theory lead to consider spacetimes $X$ such that the set of regions
$\{ Y \subset X \}$
is not directed under inclusion anymore, and the above scenario breaks down. 
In this case we should say, to be precise, that $(\cA,\jmath)$ is a precosheaf of $\rC^*$-algebras, 
and the search for Hilbert space representations may be vain (see \cite{RVL}).\\
\indent In recent times a more general notion of representation has been given for "nets" of $\rC^*$-algebras
over generic spacetimes $X$, defined in such a way that the obstacle to get coherence is encoded 
by a family 
$\{ U_{Y'Y} \}_{Y \subseteq Y'}$
of unitaries fulfilling the cocycle relations (see \cite{BR,BFM}).
These, that we simply call {\em representations}, reduce to Hilbert space representations when $X$ 
is simply connected and maintain the properties of charge composition, conjugation and covariance under 
eventual spacetime symmetries, typical of the usual sectors.\\
\indent In a previous paper (\cite{RVL}), we introduced the notion of {\em enveloping net bundle} of 
the given net of $\rC^*$-algebras $(\cA,\jmath)$; this is a net of $\rC^*$-algebras 
$(\overline{\cA},\overline{\jmath})$ 
such that any $\overline{\jmath}_{Y'Y}$ is an isomorphism, and fulfils the universal property 
of lifting any representation of $(\cA,\jmath)$. So the question of existence of representations 
is reduced to nondegeneracy of the canonical embedding 
$\e : (\cA,\jmath) \to (\overline{\cA},\overline{\jmath})$.
We call {\em injective} those nets such that $\e$ is faithful.\\
\indent In the present work we focus on nets of $\rC^*$-algebras defined over $S^1$, 
a remarkable class due to its applications in conformal quantum field theory (\cite{KL,CKL}). 
We show that any net over $S^1$ is injective, so it has faithful representations. 
Moreover, when the net is covariant under the action of $\mathrm{Diff}(S^1)$, we show the 
existence of covariant representations of any amenable subgroup of 
$\mathrm{Diff}(S^1)$.
The technique that we will use shall be the one of approximate the set of proper intervals 
of $S^1$ with finite subsets (roughly speaking, an analogue of the decomposition of 
$S^1$ as a CW-complex in the setting of partially ordered sets), then to show that the
restriction of $(\cA,\jmath)$ on these subsets is injective, and finally to prove
injectivity of the initial net performing an inductive limit.
We have postponed to Appendix \ref{Be} some rather technical computations showing
that injectivity is preserved under inductive limits, a result which plays a key r\^ole 
in the analysis of nets over $S^1$ and that, we hope, could play a similar role 
for other spacetimes too.

\section{Some preliminaries on nets.}
\label{A}

To make the present work self-contained
in this section we recall the basic properties of the main objects of our study, 
namely nets of $\rC^*$-algebras.
All the material presented here appeared in \cite{RVL};
the reader may pass to the next section whenever he is already familiar with that paper.


\subsection{Posets} 
\label{A:a}

A {\em poset} (partially ordered set) is a set endowed with an 
(antisymmetric, reflexive and transitive) order relation $\leq$.
A {\em poset morphism} is a map
$\rf : K \to K'$
such that $o \leq \tilde o$ implies $\rf(o) \leq' \rf(\tilde o)$ for all $o,\tilde o \in K$,
where $\leq'$ is the order relation of $K'$.
A {\em disjointness relation} on $K$ is a symmetric binary relation $\perp$ such that
\[
\tilde o \perp a \ , \ o \leq \tilde o 
\ \Rightarrow \
o \perp a
\ .
\]
A group $G$ is said to be a {\em a symmetry group} for $K$ whenever it acts by automorphisms on $K$,
namely 
$go \leq g \tilde o \Leftrightarrow o \leq \tilde o$ for all $g \in G$ and $o , \tilde o \in K$,
and we assume that, whenever $K$ has a disjointness relation,
$o \perp a \Leftrightarrow go \perp ga$.\smallskip

The classical covariant poset used in algebraic quantum field theory is the set of doublecones in
the Minkowski spacetime, having the inclusion as order relation, the spacelike separation as the
disjointness relation and the Poincar\'e group as the group of symmetries. We shall focus in \S \ref{E}
to the case of proper intervals in $S^1$, of interest in low dimensional quantum field theory.\smallskip 

We now give a brief description of the notion of connectedness and 
simply connectedness for posets and refer the reader to the paper \cite{RVL} 
for details\footnote{The standard way to introduce these topological notions 
makes use of a simplicial set associated to the poset. We prefer 
do not introduce this simplicial set since it will be not 
explicitly used in the present paper.}.  A poset $K$ is \emph{pathwise connected} if for any pair $a,\tilde a\in K$. 
there are two finite sequence $a_1,\ldots a_{n+1}$ and  $o_1,\ldots o_n$ of elements of $K$,  
with $a_1=a$ and $a_{n+1}=\tilde a$, 
satisfying the relations 
\[
   a_i,a_{i+1}\leq o_{i} \ , \ i=1,\ldots, n \ .  
\]
In the sequel, \emph{we will always assume that our poset is pathwise connected}. As already said, 
there is a notion of the first homotopy group $\pi^o_1(K)$ for $K$. The supscript $o$
denotes the base point in $K$  where the homotopy is calculated, however the
isomorphism class does not depend on the choice of $o$ (for this reason,
often we shall write $\pi_1(K)$ without specification of the base point).
We shall say that $K$ is {\em simply connected} whenever $\pi_1(K)$ is trivial. \\
\indent If $X$ is a space having a subbase $K$ of arcwise and simply connected open sets, 
and if $K$ is ordered under inclusion, 
then there is an isomorphism $\pi_1(X) \simeq \pi_1(K)$ (\cite{Ruz}).\smallskip

We now deal with continuous actions of symmetry groups.  
Let $G$ be topological symmetry group acting on a poset $K$. 
and $\cO(e)$ denote the set of open neighbourhoods 
of the identity of the group $G$. Then we define 
\begin{equation}
\label{Ab:0}
o \ll a \ \ \iff \ \ \exists U\in\cO(e) \ , \ go \leq   a  \ , \  \forall g\in U \ . 
\end{equation}
Now, a topological  symmetry group $G$ of $K$ is said to be a \emph{continuous symmetry group of} $K$  if 
\begin{equation}
\label{Ab:1}
\forall o\in K \ , \ \exists a\in K \ , \ \ o\ll a \ , 
\end{equation}
\begin{equation}
\label{Ab:1a}
o \ll a \ \ \Rightarrow \ \ \exists \tilde a\in K \ , \ \ o\ll \tilde a\ll a \ , 
\end{equation}
and 
\begin{equation}
\label{Ab:2}
 o\ll a_1, a_2 \ \ \Rightarrow \ \ \exists \tilde o\in K \ , \ o\ll \tilde o \ll a_1,a_2 \ . 
\end{equation}
This condition is suited for posets arising as subbases of topological $G$-spaces
and, roughly speaking, encodes the idea that the sets 
$\{ go , g \in U \}$, $U \in \cO(e)$,
yield a neighbourhood system for $o$.
Double cones in Minkowski spacetime and the open intervals of $S^1$ are examples of posets 
acted upon continuously (in the above sense) by the Poincar\'e group and the 
$\mathrm{Diff}(S^1)$ group respectively. 
In these cases it is easely seen that 
$a \ll o$ is equivalent to the condition that the closure of $a$ is contained in $o$. 
Note, in addition, that the above conditions are always verified when 
the symmetry group $G$ has the  discrete topology. 
\begin{remark}
The notion of a continuous symmetry group of a poset introduced in the present paper
is different from that used in \cite{RVL}. We prefer this new notion of continuity  because 
it involves in its definition only the poset and the group. The older, instead,  
involved the simplicial set associated to the poset.  
In Appendix \ref{AppB} we shall prove that the new notion of continuity 
is stronger than the older one,  
so that all the results obtained in that paper continue to hold. 
\end{remark}

\subsection{Nets of $\rC^*$-algebras}
\label{Ba}

A {\em net of $\rC^*$-algebras} over the poset $K$ is given by a family 
$\cA := \{ \cA_o \}_{o \in K}$
of unital $\rC^*$-algebras (called the {\em fibres}), and a family 
$\jmath := \{ \jmath_{\tilde oo} : \cA_o \to \cA_{\tilde o} , o \leq \tilde o \}$
of unital monomorphisms (called the {\em inclusion maps}) fulfilling the {\em net relations}
\[
\jmath_{o' \tilde o} \circ \jmath_{\tilde oo} = \jmath_{o'o}
\ \ , \qquad  o \leq \tilde o \leq o'
\ .
\]
In the sequel we shall denote a net of $\rC^*$-algebras by $(\cA,\jmath)_K$.
When every $\jmath_{\tilde oo}$ is an isomorphism we say that $(\cA,\jmath)_K$
is a {\em $\rC^*$-net bundle} and, to be short, we write 
$\jmath_{o \tilde o} := \jmath_{\tilde oo}^{-1}$, $\forall o \leq \tilde o$.
The {\em restriction of} $(\cA,\jmath)_K$ over $S \subset K$ is given by the
same families restricted to elements of $S$ and is denoted by $(\cA,\jmath)_S$.\smallskip

Clearly, the definition of net can be given for other categories and in particular
we shall use the one of Hilbert spaces, especially the case of Hilbert net bundles
(whose net structure is given by unitary operators).\\
\indent In all the cases of interest $K$ shall be a subbase for the topology of a space, 
in general not directed, so if we would use the correct terminology
in the setting of algebraic topology we should use the term 
{\em precosheaf of $\rC^*$-algebras}; neverthless, we prefer to maintain the 
usual term {\em net}, since it is standard in algebraic quantum field theory.\smallskip

A {\em morphism} of nets is written 
\[
(\pi,\rf) : (\cA,\jmath)_K \to (\cB,i)_P \ ,
\]
where $\rf : K \to P$ is a poset morphism and 
$\pi := \{ \pi_o : \cA_o \to \cB_{\rf(o)} \}$
is a family of unital morphisms such that 
$i_{\rf(\tilde o),\rf(o)} \circ \pi_o = \pi_{\tilde o} \circ \jmath_{\tilde oo}$, 
$\forall o \leq \tilde o$. When $\rf$ is the identity we shall write $\pi$ instead
of $(\pi,id_K)$. We say that $(\pi,\rf)$ is \emph{faithful on the fibres} if $\pi_o$ is 
a monomorphism  for any $o$; it is an \emph{isomorphism} when both $\pi$ and $\rf$ are isomorphisms.
We say that a net is {\em trivial} if it is isomorphic to 
the {\em constant net} $(\cC,id)_K$, where $\cC_o \equiv \cF$ for a fixed
$\rC^*$-algebra $\cF$ and any $id_{\tilde oo}$ is the identity of $\cF$.\smallskip


The structures that we introduce in the following lines are familiar in the
setting of quantum field theory, and reflect  Poincar\'e
(M\"obius) symmetry and Einstein causality respectively.
If $G$ is a continuous symmetry group of $K$, then we say that the net $(\cA,\jmath)_K$ is $G$-\emph{covariant}
whenever there are isomorphisms 
$\alpha^g_o : \cA_o \to \cA_{go}$, $\forall o \in K$, $g \in G$,
such that
\[
\alpha^g_{\tilde o} \circ \jmath_{\tilde oo} = \jmath_{g \tilde o , go} \circ \alpha^g_o
\ \ , \ \
\alpha^h_{go} \circ \alpha^g_o = \alpha^{hg}_o
\ \ , \ \ o \leq \tilde o \in K
\ , \
g,h \in G
\ , 
\]
and fulfilling the following \emph{continuity condition}: if 
$\{ g_\lambda \} \subset G$  is a net converging to $e$, then  for any $o\in K$ 
there exists $a\gg o$ 
and an  index $\lambda_a$ such that $g_\lambda o \leq a$, $\forall \lambda \geq \lambda_a$,
and
\begin{equation}
\norm{\jmath_{a\, g_\lambda o}\circ \alpha^{g_\lambda}_o(A)-\jmath_{ao}(A)} \to 0
\ \ , \ \ A\in\cA_o
\ .
\end{equation}
If $(\cA,\jmath,\alpha)_K$ and $(\cB,i,\beta)_K$ are $G$-covariant nets,   a morphism $(\pi,\rf):(\cA,\jmath,\alpha)_K\to(\cB,i,\beta)_K$
is said to be $G$-\emph{covariant} whenever 
\[
 \rf(go)= g\rf(o) \ \ , \ \   \pi_{go}\circ \alpha^g_o = \beta^g_{\rf (o)}\circ \pi_o \ , \qquad o\in K \ , \ g\in G \ .  
\] 
Finally, when $K$ has  a causal  disjointness relation $\perp$, 
we say that the net $(\cA,\jmath)_K$ is {\em causal} whenever 
\[
[ \ \jmath_{a o}(t) \ , \ \jmath_{a \tilde o}(s) \ ] = 0 \ , \qquad 
 o \perp \tilde o \ , \ \ o,\tilde o\leq a 
\ , 
\]
where  $t \in \cA_{o}$ and  $s \in \cA_{\tilde o}$.


%
\subsection{The enveloping net bundle and injectivity}
\label{Bd}

The importance of net bundles in the analysis of nets resides in the following fact. 
Let $(\cA,\jmath)_K$ be a net bundle. Since the inclusion maps $\jmath$ are isomorphism, they 
induce,   for any $o\in K$, an action, \emph{the holonomy action},  
\begin{equation}
\label{Ba:1}
\jmath_* : \pi^o_1(K) \to {\bf aut}\cA_o \ ,
\end{equation}
of the homotopy group $\pi^o_1(K)$ into the fibre $\cA_o$. The $\rC^*$-dynamical system $(\cA_o,\pi_1^o(K),\jmath_*)$ is unique up to isomorphism
at varying of $o$ in $K$ and is a complete invariant of the net bundle since  the
net bundle can be reconstructed (up to isomorphism) starting from the $\rC^*$-dynamical system. 
The net bundle $(\cA,\jmath)_K$ is trivial 
if and only if $\jmath_*$ is the trivial action. We shall refer to $(\cA_o,\pi_1^o(K),\jmath_*)$
as the \emph{holonomy dynamical system} of the net bundle. \\ 
\indent On these grounds it is crucial to understand whether and how a net can  be 
embedded into a net bundle. \smallskip

It turns out that any net of $\rC^*$-algebras can be embedded into a $\rC^*$-net bundle. 
To be precise, the \emph{enveloping net bundle} 
of  a net of $\rC^*$-algebras $(\cA,{\jmath})_K$ is a  
$\rC^*$-net bundle by $(\overline{\cA},\overline{\jmath})_K$, which comes
equipped with a morphism  $\e : (\cA,\jmath)_K \to (\overline{\cA},\overline{\jmath})_K$, called
the {\em canonical embedding}, satisfying the following remarkable universal properties:
given morphisms with values in $\rC^*$-net bundles,
\[
(\varphi,\rh) , (\theta,\rh) : (\overline{\cA} , \overline{\jmath})_K \to (\cC,y)_P
\ \ , \ \
(\psi,\rf) : (\cA,\jmath)_K \to (\cB,\imath)_S
\ ,
\]
we have
\begin{equation}
\label{Bd:11}
\left\{
\begin{array}{ll}
(\varphi,\rh) \circ \e = (\theta,\rh) \circ \e  \ \Rightarrow \ \varphi = \theta \ ,
\\
\exists !  \ (\psi^\uparrow,\rf) \ {\mathrm{such \ that}} \ (\psi^\uparrow,\rf) \circ \e = (\psi,\rf) \ ,
\end{array}
\right.
\end{equation}
where $(\psi^\uparrow,\rf)$ is the pullback
\begin{equation}
\label{Bd:11a}
(\psi^\uparrow,\rf) : (\overline{\cA},\overline{\jmath})_K \to (\cB,\imath)_S
\ .
\end{equation}
These properties characterize the enveloping net bundle, that is, 
it is the unique, up to isomorphism,  
$\rC^*$-net bundle satisfying the above relations, and this leads to the  following classification:  a net of $\rC^*$-algebras is {\em degenerate} if its enveloping net bundle 
vanishes, and is {\em nondegenerate} otherwise. A nondegenerate 
net of $\rC^*$-algebras is {\em injective} if the canonical embedding is a monomorphism.
\begin{remark}
(1) When $K$ is simply connected the  $(\overline{\cA},\overline{\jmath})_K$ is a trivial $\rC^*$-net bundle with fibres
isomorphic to the Fredenhagen universal $\rC^*$-algebra of $(\cA,\jmath)_K$ (see \cite{Fre}).\\[5pt]
(2) If $G$ is a continuous symmetry group, then the enveloping net bundle of a  
$G$-covariant net is $G$-covariant as well. 
\end{remark}

We can now state the functoriality property and its relation with injectivity: 
for any morphism 
$(\rho,\rf):(\cA,\jmath)_K \to (\cD,k)_P$ there exists a morphism 
$(\overline{\rho},\rf) : 
(\overline{\cA},\overline{\jmath})_K \to (\overline{\cD},\overline{k})_P$ which fulfils the 
the relation
\[
(\overline{\rho},\rf)\circ \e = \tilde \e \circ (\rho,\rf) \ , 
\]
where $\e$ and $\tilde \e$ are, respectively, the canonical embeddings of 
the nets  $(\cA,\jmath)_K$ and $(\cD,k)_P$. This makes the assignment of the enveloping net bundle a functor. 
%
%
\begin{remark}
\label{Bd:12}
(1) Note that if $(\rho,\rf)$ is faithful on the fibres and $(\cD,k)_P$ is injective, 
then $(\cA,\jmath)_K$ is injective too. \\[5pt]
(2) The functor assigning the enveloping net bundle preserves inductive limits (Prop.\ref{Beb:3}). 
\end{remark}

\subsection{States and representations}
\label{C}

A \emph{state} of a net of $\rC^*$-algebras $(\cA,{\jmath})_K$ is a family of states of $\rC^*$-algebras
$\omega := \{ \omega_o : \cA_o \to \bC \ , o \in K \}$ 
fulfilling
\begin{equation}
\label{Ca:1}
\omega_o = \omega_{a} \circ {\jmath}_{ao}
\ \ , \ \ 
o \leq a
\ .
\end{equation}
It turns out that the set of states 
of a $\rC^*$-net bundle $(\cA,\jmath)_K$ is in one-to-one correspondence with 
the set of invariant states of the associated holonomy dynamical system 
$(\cA_o,\pi^o_1(K),\jmath_*)$. Since in a $\rC^*$-dynamical system having amenable group  
invariant states always exist, 
we conclude that when the fundamental group of $K$ is amenable then any
nondegenerate net has states; in fact we can compose states of the enveloping net bundle
with the canonical embedding.
If $(\cA,\jmath,\alpha)_K$ is $G$-covariant net, then a state  of the net $\varphi$
is said to be $G$-\emph{invariant} whenever
\[
\varphi_{go} \circ \alpha^g_o := \varphi_o
\ \ , \ \
\forall o \in K
\ , \
g \in G
\ .
\]
The next result concerns the existence of $G$-invariant states and is proved in \cite{RVL}:
\begin{proposition}[\cite{RVL}]
\label{Ca:6} 
Let $G$ be amenable. Then:
(i)  Any $G$-covariant $\rC^*$-net bundle having states has $G$-invariant states. 
(ii) If $\pi_1(K)$ is amenable, then any nondegenerate $G$-covariant net 
over $K$ has $G$-invariant states.
\end{proposition}

A \emph{representation} of the net $(\cA,{\jmath})_K$ is given by a pair
$(\pi,U)$, where $\pi := \{ \pi_o : \cA_o \to \cB(\cH_o)  \}$ is a family
of representations and 
$U := \{ U_{\tilde oo} : \cH_o \to \cH_{\tilde o} , o \leq \tilde o \}$ is a family of unitary
operators fulfilling the relations
\begin{equation}
\label{Cb:2}
U_o \in ( \pi_o , \pi_{\tilde o} \circ \jmath_{\tilde oo} )
\ \ , \ \
U_{o' \tilde o} \circ U_{\tilde oo} = U_{o'o}
\ \ , \ \
\forall o \leq \tilde o  \leq o'
\ .
\end{equation}
We call $U$ the family of {\em inclusion operators}.
We say that $(\pi,U)$ is {\em faithful} whenever $\pi_o$ is faithful
for any $o \in K$, and that $(\pi,U)$ is a {\em Hilbert space representation}
whenever any $U_{\tilde oo}$ is the identity on a fixed Hilbert space
(and in this case we write $(\pi,\mathbbm{1})$).

It follows from (\ref{Cb:2}) that the pair $(\cH,U)_K$, $\cH := \{ \cH_o \}$,
is a Hilbert net bundle in the sense of the previous section. Using the
adjoint action we obtain the $\rC^*$-net bundle $( \cB \cH , \ad U )_K$,
so that $(\pi,U)$ can be regarded as a morphism
\[
\pi : (\cA,\jmath)_K \to ( \cB \cH , \ad U )_K \ .
\]
Hilbert space representations correspond, in essence, to morphisms with values
in trivial nets. When $K$ is simply connected any $( \cB \cH , \ad U )_K$ is
trivial, so we have only Hilbert space representations. When $K$ is not
simply connected it is very easy to give examples of nets having faithful
representations but no Hilbert space representations (see \cite{RVL}),
and this is the reason why it is convenient to use the more general definition.
In algebraic quantum field theory it is customary to use Hilbert spaces
representations, also because the usual background is the Minkowski 
spacetime that is simply connected. In curved spacetimes and in $S^1$
it is of interest to give results stating the existence of (possibly)
faithful representations, and this is the motivation of our work.\smallskip

Let us focus for a moment on $\rC^*$-net bundles.  There exists a 
one-to-one correspondence between
representations $(\pi,U)$ of $(\cA,\jmath)_K$ and covariant representations $(\pi_o,U_*)$ of
the holonomy dynamical system $(\cA_o,\pi^o_1(K),\jmath_*)$. In particular 
$U_*$, which  is a unitary representation of the fundamental group of $K$, is nothing but that 
the holonomy of the Hilbert net bundle $(\cH,U)_K$ (see \ref{Ba:1}).
Since any $\rC^*$-dynamical system 
has faithful covariant representations, we conclude that any $\rC^*$-net bundle
has faithful representations.\\
\indent We now return to the general case in which $(\cA,\jmath)_K$ is a net of $\rC^*$-algebras.
Using the pullback (see (\ref{Bd:11a})), we see that any representation $(\pi,U)$ of
$(\cA,\jmath)_K$ extends to a representation $(\pi^\uparrow,U)$ of 
$(\overline{\cA},\overline{\jmath})_K$ and this yields a one-to-one correspondence
between representations of $(\cA,\jmath)_K$ and those of its enveloping net bundle.
Thus, as $\rC^*$-net bundles are faithfully represented, we conclude that 
{\em a net of $\rC^*$-algebras is injective if, and only if, it has faithful representations}.\smallskip

Let $G$ be a continuous symmetry group of $K$ and $(\cA,\jmath,\alpha)_K$ 
a $G$-covariant net. A \emph{$G$-covariant representation} of $(\cA,\jmath,\alpha)_K$ is a
representation $(\pi,U)$ of $(\cA,\jmath)_K$ with a strongly continuous family $\Gamma$ 
of unitaries
$\Gamma^g_o : \cH_o \to \cH_{go}$, $g \in G$, $o \in K$, 
such that
\[
\Gamma^h_{go} \circ \Gamma^g_o = \Gamma^{hg}_o
\ \ , \ \
\ad \Gamma^g_o \circ \pi_o = \pi_{go} \circ \alpha^g_o
\ \ , \ \
\Gamma^g_{\tilde o} \circ U_{\tilde oo} = U_{\tilde go \ go} \circ \Gamma^g_o
\ ,
\]
for all $g,h \in G$, $o \leq \tilde o \in K$.
With the term \emph{strongly continuous} we mean the following property:
if $\{ g_\lambda \} \subset G$ converges to the identity, 
then for any  $o\in K$  there exists $a \gg o$  and an index 
$\lambda_a$ such that $g_{\lambda}o\leq a$ for any $\lambda\geq \lambda_a$ and   
\begin{equation}
\label{Cb:12a}
\| U_{a\, g_\lambda o}\, \Gamma^{g_\lambda}_o \Omega -\Omega\| \to 0 \ ,\qquad \forall \Omega\in\cH_a \ . 
\end{equation}
We have the following result (see, as usual, \cite{RVL}):
\begin{proposition}
\label{Cb:13}
Let $K$ be a poset with amenable fundamental group and 
$G$ an amenable continuous symmetry group of $K$. 
Then every injective, $G$-covariant net of $\rC^*$-algebras over $K$ has
strongly continuous $G$-covariant representations.
\end{proposition}


\section{Nets over  $S^1$}
\label{E}

Let $\cI$ be the poset formed by the set of connected, open intervals 
of $S^1$ having closure $cl(o)$ properly contained in $S^1$, ordered by inclusion; 
that is, $o\leq a$ if, and only if, $o\subseteq a$. 
The homotopy group of this poset is $\mathbb{Z}$, since $\cI$ 
is a base for the topology of $S^1$. By a \emph{net of $\rC^*$-algebras over $S^1$} we mean
a net of $\rC^*$-algebras over $\cI$. \\ 
\indent On $\cI$ there is a natural causal disjointness relation: $o\perp a$ if, and only if, 
$o\cap a=\emptyset$. Important symmetries for nets over $S^1$ are given by $\mathrm{Diff}(S^1)$ 
or the M\"obius subgroup. These groups act continuously on $S^1$ and, hence,
on the poset $\cI$ as well, according to (\ref{Ab:1},\ref{Ab:1a},\ref{Ab:2}). These groups are non-amenable.
However there are important amenable subgroups: the rotation group, the semidirect product 
of the translations and the dilations. \\  
\indent In the present section we show that any net of $\rC^*$-algebras over $S^1$ is injective, 
so it admits faithful representations. As a consequence any such a net has states and, 
if the net is covariant under $\mathrm{Diff}(S^1)$, states which are 
invariant under any amenable subgroup of $\mathrm{Diff}(S^1)$. \\ 
\indent To prove injectivity, we shall follow the strategy of 
finding a family of finite, "approximating" subposets of $\cI$, that we call {\em cylinders}, 
where injectivity can be established.

\subsection{Idea and scheme of the proof}
\label{Da}
A strategy for proving injectivity  is suggested by the analysis 
of inductive systems of nets (see in appendix).  
Let $(\cA,\jmath)_K$ be a net of $\rC^*$-algebras over a poset $K$.
Assume that there is a family $\{K^\alpha\}$ of subsets of $K$ satisfying 
the following conditions:
\begin{enumerate}
\item the family $\{K^\alpha\}$
is upward directed under inclusion and  $K$ is the inductive limit poset of  $\{K^\alpha\}$
(each $K^\alpha$ equipped with the order relation inherited by $K$); 
\item the net $(\cA,\jmath)_{K^\alpha}$, that is the restriction  of $(\cA,\jmath)_K$ to $K^\alpha$, is injective for any $\alpha$;
\end{enumerate}
Condition 1  says that the net $(\cA,\jmath)_K$ is  the inductive limit of the system $(\cA,\jmath)_{K^\alpha}$. Condition 2 implies the injectivity of the inductive limit
net (Theorem \ref{Beb:10}). \smallskip

Although  it is a hard problem, even impossible in some cases,  to find  the right family of subsets of a poset 
where injectivity can be established, this problem, in the case of $S^1$, can be fully solved. We briefly explain how.    
Given a net $(\cA,\jmath)_{\cI}$ we will find a sequence $\{\cI_N\}$,  $N\in\mathbb{N}$, 
of subsets of $\cI$
satisfying the condition 1. 
Concerning condition 2.,  first we will construct (using $\cI_N$)  a finite poset $P_N$ 
and show that the net $(\cA,\jmath)_{\cI_N}$ embeds, faithfully on the fibres, into a suitable net over $P_N$. Secondly, we shall show that any net of $\rC^*$-algebras over $P_N$ 
is injective. These facts imply that the restrictions 
$(\cA,\jmath)_{\cI_N}$ are injective for any $N$. \\
\indent The proof that any net over $P_N$ is injective relies on the isomorphism  
between $P_N$ and  an abstract poset $C_N$, called the $N$-cylinder, and on the fact 
that any net over this poset is injective.

\subsection{Nets over the $N$-Cylinder}
\label{Db}
We now introduce a class $\{ C_N , N \in \bN \}$ of finite posets. 
These are of interest for two reasons; the first one is that any net of 
$\rC^*$-algebras over some $C_N$ is injective, and the second one is that,
as we shall explain in the following, each $C_N$ arises from a suitable 
simplicial approximation of the circle. \bigskip 

As we shall see soon $C_N$ 
can be seen as a lattice of $N^2$ elements on a cylinder of finite height. 
To deal with periodicity 
we shall use the equivalence relation $mod$ $N$ with the following convention: 
we choose as representative elements of the 
classes associated with the equivalence relation $mod$ $N$ the numbers $1, 2,\ldots, N$. 
So, for instance, for $N=4$  we have, $(0)_4=4$, $(-1)_4=3$, $(5)_4=1$, etc... . \\
\indent Using this convention, elements of the $N$-\emph{cylinder} $C_N$ are pairs $(i,l)$ with $i,l\in \{1,\ldots, N\}$. 
We shall think of $C_N$ as a matrix whose rows and columns are indexed by $l$ and $i$ 
respectively. The order relation is defined 
inductively, as follows: given an element $(i,l)$ of the $l$-row, with $l < N$, 
it has only two majorants in the $(l+1)$-row, given by
\begin{equation}
\label{Db:1}
 (i,l) < (i,l+1) \ , \ \ \ \ (i,l) < ((i-1)_N,l+1) \ . 
\end{equation}
Finally, the relation among $(i,l)$ and that of the $(l+t)$-rows with $t>1$ 
is obtained by transitivity. For $N=4$, the poset $C_4$ is represented by the following 
diagram,
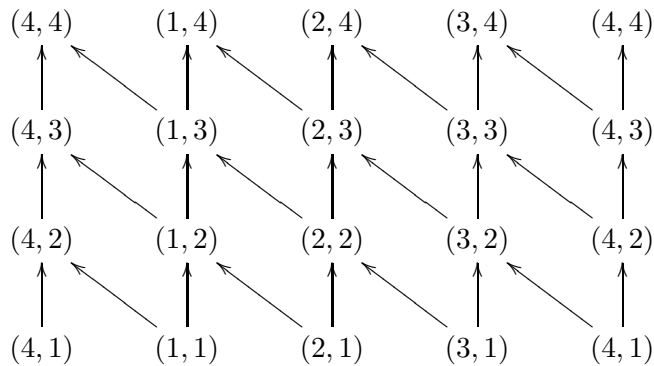
\[
\xymatrix{
(4,4)&         (1,4)&               (2,4)&              (3,4) &               (4,4)\\ 
(4,3)\ar[u]  & (1,3)\ar[u]\ar[lu] & (2,3)\ar[u]\ar[lu]& (3,3) \ar[u] \ar[lu]& (4,3)\ar[u]\ar[lu] \\
(4,2)\ar[u] &  (1,2)\ar[u]\ar[lu] & (2,2)\ar[u]\ar[lu]& (3,2)\ar[u]\ar[lu] &  (4,2)\ar[u]\ar[lu] \\ 
(4,1)\ar[u]&   (1,1)\ar[u]\ar[lu] & (2,1)\ar[u]\ar[lu]& (3,1)\ar[u]\ar[lu] &  (4,1)\ar[u]\ar[lu] 
}
\]
where, for simplicity, the first column is the repetition of the last.  
The order relation is represented by an arrow from the 
smaller element to the greater one. So,   $C_N$ has $N$ maximal elements, 
those belonging to the $N$-row, and $N$ minimal elements, those belonging to 
the $1$-row. \smallskip

The rest of the section is devoted to proving that any net of $\rC^*$-algebras
$(\cA,\jmath)_{C_N}$  admits a faithful representation, a  property  equivalent to 
injectivity (see \S \ref{C}).  To this end we shall use  
an idea of Blackadar \cite{Bla}.  
Consider the algebras associated with the maximal elements: $\cA_{(i,N)}$. 
Take a cardinal $\kappa$ greater than the cardinality of any such an algebra. 
Let $\rho_i$ denote the tensor product of
the universal representation of the algebra $\cA_{(i,N)}$ and of $1_{\kappa}$. 
Then define 
\begin{equation}
\label{Db:2}
\pi_{(i,l)}:= \rho_{i}\circ \jmath_{(i,N)(i,l)} \ , \qquad l=1,2,\ldots, N \ .
\end{equation}
In words, the representation of the algebras associated with elements of the $i$-column
is obtained by restricting $\rho_i$ to such algebras. In particular $\pi_{(i,N)}=\rho_i$. 
In the case of $N=4$ we will label the columns of the above diagram as follows 
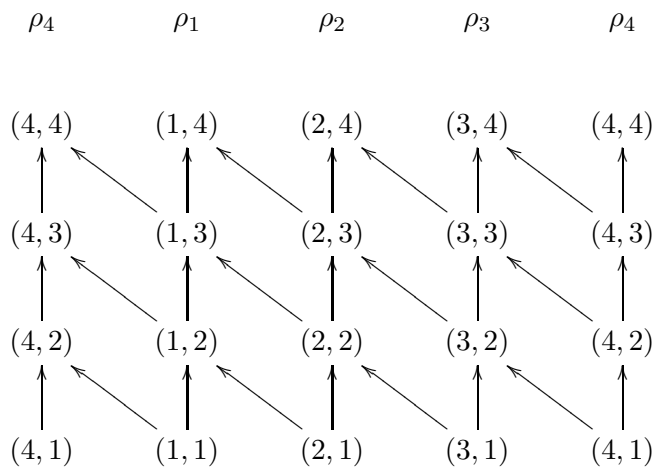
\[
\xymatrix{
\rho_4& \rho_1 & \rho_2 & \rho_3 & \rho_4 \\
(4,4)&         (1,4)&               (2,4)&              (3,4) &               (4,4)\\ 
(4,3)\ar[u]  & (1,3)\ar[u]\ar[lu] & (2,3)\ar[u]\ar[lu]& (3,3) \ar[u] \ar[lu]& (4,3)\ar[u]\ar[lu] \\
(4,2)\ar[u] &  (1,2)\ar[u]\ar[lu] & (2,2)\ar[u]\ar[lu]& (3,2)\ar[u]\ar[lu] &  (4,2)\ar[u]\ar[lu] \\ 
(4,1)\ar[u]&   (1,1)\ar[u]\ar[lu] & (2,1)\ar[u]\ar[lu]& (3,1)\ar[u]\ar[lu] &  (4,1)\ar[u]\ar[lu] 
}
\]
Since universal representations and the inclusion 
maps are faithful,  any representation $\pi_{(i,l)}$ is faithful. \\
\indent We now define the inclusion operators. 
We proceed by defining the inclusion operators from a $l$-row to $(l-1)$-row, starting from $l=N$;
the others will be obtained by composition. To this end we note that the maximal element $(i,N)$
has two minorants in the $(N-1)$-row, $(i,N-1)$ and $((i+1)_N, N-1)$. 

According to Definition \ref{Db:2} we take the identity operator $\mathbbm{1}$ as the inclusion
operator from $(i,N-1)$ and $(i,N)$, because these two elements belong to the same column $i$. 
Concerning the inclusion operator from $((i+1)_N, N-1)$ to $(i,N)$, 
%
%
the representations $\pi_{(i,N)}\circ \jmath_{(i,N)\,((i+1)_N,N-1)}$ 
and $\pi_{((i+1)_N,N-1)}$ are unitarily equivalent, because  
they are unitarily equivalent to tensor product of the universal representation of the algebra 
$\cA_{(i+1)_N, N-1)}$ and   $1_\kappa$ 
(see \cite{Bla}). So, there is a unitary operator $V_{i,(i+1)_N}$ such that 
\begin{equation}
\label{Db:3}
  V_{i,(i+1)_N}\, \pi_{((i+1)_N,N-1)} = \pi_{(i,N)}\circ \jmath_{(i,N)\,((i+1)_N,N-1)}
\, V_{i,(i+1)_N} \ . 
\end{equation}
So we take $V_{i,(i+1)_N}$  as inclusion operator from $((i+1)_N, N-1)$ to $(i,N)$; the reason way 
it is labelled only by the column indices will become 
clear soon. Given an element $(i,l)$, with $1<l<N$,   
consider the minorants $(i,l-1)$ and $((i+1)_N, l-1)$. As before we take 
the identity $\mathbbm{1}$ as the inclusion operator from $(i,l-1)$ and $(i,l)$, since they belong
to the same column. But the important fact is that we may 
take the same operator  $V_{i,(i+1)_N}$ satisfying 
equation (\ref{Db:3}) as inclusion operator from $((i+1)_N, l-1)$ to $(i,l)$. In fact 
by using equation (\ref{Db:3}) and Definition (\ref{Db:2}) we have 
\begin{align*}
V_{i,(i+1)_N} & \, \pi_{((i+1)_N,l-1)}   = \\
& =   V_{i,(i+1)_N}\, \rho_{(i+1)_N}\circ \jmath_{((i+1)_N,N)\, ((i+1)_N,l-1)}   \\
& = V_{i,(i+1)_N}\, \rho_{(i+1)_N} \circ \jmath_{((i+1)_N,N)\, ((i+1)_N,N-1)} \circ 
\jmath_{((i+1)_N,N-1)\,  ((i+1)_N,l-1)}   \\ 
& = V_{i,(i+1)_N}\, \pi_{((i+1)_N,N-1)}\circ \jmath_{((i+1)_N,N-1)\,  ((i+1)_N,l-1)}  \\ 
& = \pi_{(i,N)}\circ\jmath_{(i,N)\,((i+1)_N,l-1)} \, V_{i,(i+1)} \\ 
& = \pi_{(i,N)}\circ\jmath_{(i,N)\,(i,l)} \circ \jmath_{(i,l)\,((i+1)_N,l-1)}\, V_{i,(i+1)}\\
& = \pi_{(i,l)} \circ \jmath_{(i,l)\,((i+1)_N,l-1)} \, V_{i,(i+1)} \ ,  
%
\end{align*}
where we have used the relations  $((i+1)_N,N-1) >  ((i+1)_N,l-1)$ 
and $(i,N) > (i,l)$ for $N>l>1$. Hence 
\begin{equation}
\label{Db:4}
V_{i,(i+1)_N} \, \pi_{((i+1)_N,l-1)}   = \pi_{(i,l)} \circ \jmath_{(i,l)\,((i+1)_N,l-1)} \, V_{i,(i+1)} \ .
\end{equation}
This choice, in the case of $C_4$, corresponds to  the diagramme 
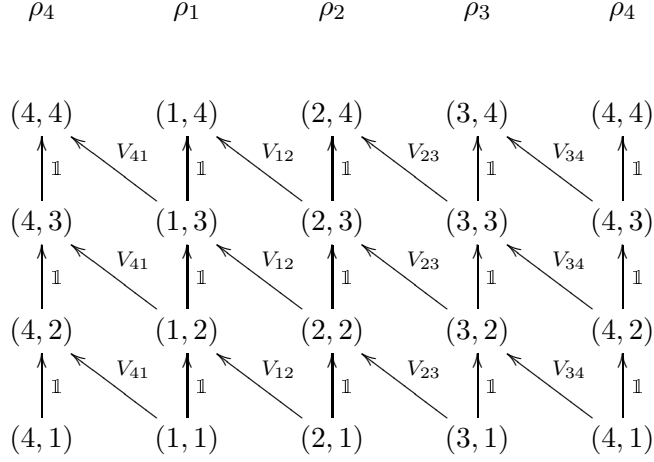
\[
\xymatrix{
\rho_4& \rho_1 & \rho_2 & \rho_3 & \rho_4 \\
 (4,4) & (1,4)& (2,4)&  (3,4) & (4,4)\\ 
 (4,3)\ar[u]_{\mathbbm{1}}  & (1,3)\ar[u]_{\mathbbm{1}} \ar[lu]_{V_{41}} & (2,3)\ar[u]_{\mathbbm{1}} \ar[lu]_{V_{12}}& (3,3) \ar[u]_{\mathbbm{1}}  \ar[lu]_{V_{23}}& (4,3)\ar[u]_{\mathbbm{1}} \ar[lu]_{V_{34}} \\
(4,2)\ar[u]_{\mathbbm{1}} & (1,2)\ar[u]_{\mathbbm{1}}\ar[lu]_{V_{41}} & (2,2)\ar[u]_{\mathbbm{1}}\ar[lu]_{V_{12}}& (3,2)\ar[u]_{\mathbbm{1}}\ar[lu]_{V_{23}} & (4,2)\ar[u]_{\mathbbm{1}}\ar[lu]_{V_{34}} \\ 
(4,1)\ar[u]_{\mathbbm{1}}& (1,1)\ar[u]_{\mathbbm{1}}\ar[lu]_{V_{41}} & (2,1)\ar[u]_{\mathbbm{1}}\ar[lu]_{V_{12}}& (3,1)\ar[u]_{\mathbbm{1}}\ar[lu]_{V_{23}} & (4,1)\ar[u]_{\mathbbm{1}}\ar[lu]_{V_{34}} 
}
\]
Finally, consider a generic inclusion $(i_k,l_k) > (i_1,l_1)$. This inclusion 
can be obtained by a composition  
\begin{equation}
\label{Db:5}
(i_1,l_1) \ < \ (i_2,l_2) \ < \ \cdots \ < (i_{k-1}, l_{k-1}) \ < \ (i_k,l_k) \ ,  
\end{equation} 
of the generators (\ref{Db:1}). 
Given such a composition we define the inclusion operator 
\begin{equation}
\label{Db:6}
V_{(i_k,l_k) (i_1,l_1)} := 
V_{(i_k,l_k)(i_{k-1},l_{k-1})}\,
V_{(i_{k-1},l_{k-1}) (i_{k-2},l_{k-2})}
\,\cdots\,  
V_{(i_2,l_2)(i_{1},l_{1})}   
\ , 
\end{equation}
where the inclusion operators for the generators of $C_N$ are defined according 
to the above prescriptions. However note that  for $l_k \geq l_1+2$ 
the inclusion $(i_k,l_k) > (i_1,l_1)$ can be obtained by different compositions 
of the generators (\ref{Db:1}). For instance, for $C_4$, the inclusion $(3,1) < (2,3)$ 
can be obtained either 
\[
(3,1) \ < \ (3,2) \ < \ (2,3) \ , 
\]
or 
\[
 (3,1) \ < \  (2,2) \ < \  (2,3) \ .
\]
However, the definition (\ref{Db:6}) does not depend on the chosen composition.  
Because  in any  such composition  the following ordered sequence  of transitions from a column 
to the preceding one ($mod$ $N$) must be present: 
\begin{equation}
\label{Db:7}
 i_1 \to (i_{1}-1)_N \to\cdots \to (i_{k}+1)_N \to i_{k}  \ ,
\end{equation}
and no other. So the difference between two compositions of inclusions leading to 
$(i_k,l_k) > (i_1,l_1)$ depends on how  inclusions preserving the column index 
are inserted between the elements of the sequence. However the inclusion
operator, associated to inclusions which preserve the column index, is the identity. 
Since the inclusion operators associated with inclusions of the form 
$(i,l)< ((i-1)_N,l+1)$ depend only on the column index, the definition 
(\ref{Db:6}) does not depend on the chosen path. 
Thus the pair $(\pi,V)$ is a faithful representation of $(\cA,\jmath)_{C_N}$ (see \ref{Cb:2}),
and in conclusion we have the following 
\begin{proposition}
\label{Db:8}
Any net of $\rC^*$-algebras over $C_N$ is injective. 
\end{proposition}
%
%
%
%
%
%

\subsection{Finite approximations of $S^1$, and injectivity} 
\label{Ea}
Following  the strategy outlined in \S \ref{Da} we prove that any net of $\rC^*$-algebras 
$(\cA,\jmath)_\cI$ over $S^1$ is injective.

\paragraph{The subsets $\cI_N$.} Let $\{x_n\}$ be a dense sequence of points in $S^1$. Define 
\begin{equation}
\label{Ea:2def}
 \cI_N:= \cup^N_{i=1} \cI_{x_i} \ , \qquad N\in\mathbb{N}  \ , 
\end{equation}
where $\cI_{x}:= \{ o\in \cI \ | \ x\not\in cl(o)\}$, for a point $x$ of $S^1$, and  
$cl(o)$ denotes the closure of the interval $o$.  Note that $\cI_{x}$ is, 
as a subposet of $\cI$, upward directed. 
For $N\geq 2$ the poset $\cI_N$ is a base of neighbourhoods for the topology of $S^1$, 
so its homotopy group is $\bZ$. \\
\indent The family $\{\cI_N\}$ satisfy
condition 1. outlined in \S \ref{Da}. To this end we first observe that note that $\cI_N\subset\cI_{N+1}$ for any $N\in\mathbb{N}$. So given a net of $\rC^*$-algebras 
$(\cA,\jmath)_{\cI}$ consider the restrictions $(\cA,\jmath)_{\cI_N}$ for any $N$, and 
for any inclusion $N\leq M$ define 
\[
\left\{
\begin{array}{ll}
\mathrm{i}^{M,N}(o):=o \ , &  o\in\cI_N \ ,  \\ 
\iota^{M,N}_o(A):=A \ ,    &  o\in \cI_N \ , \ \ A\in\cA_o \ ,
\end{array}
\right.
\]
giving unital monomorphisms 
$(\iota^{M,N},\mathrm{i}^{M,N}):(\cA,\jmath)_{\cI_N}\to(\cA,\jmath)_{\cI_{M}}$. 
\begin{lemma}
\label{Ea:2a}
Given a net  of $\rC^*$-algebras $(\cA,\jmath)_{\cI}$ over $S^1$, then 
\begin{itemize}
\item[(i)] $\big\{(\cA,\jmath)_{\cI_N}, (\iota^{M,N},\mathrm{i}^{M,N})\big\}_{\bN}$ is an 
inductive system of  nets of $\rC^*$-algebras whose
linking morphisms $(\iota^{M,N},\mathrm{i}^{M,N})$ are monomorphisms.
\item[(ii)] $(\cA,\jmath)_{\cI}$ is isomorphic to the inductive limit 
of $\big\{(\cA,\jmath)_{\cI_N}, (\iota^{M,N},\mathrm{i}^{M,N})\big\}_{\bN}$. 
\end{itemize}
\end{lemma}
\begin{proof}
$(i)$ easily follows from the definition  of inductive system  (\ref{Bea:0a}). $(ii)$ 
Once we have shown that 
$\cI$ is the inductive limit poset of $(\cI_N,\mathrm{i}^{M,N})_{\bN}$ (see \S \ref{Bea}), 
the proof follows from the definition of  
$(\cA,\jmath)_{\cI_N}$ and from the universal property of inductive limits Prop.\ref{Bea:6}, 
To this end it is enough to observe that since $\{x_n\}$ is, by assumption, dense in  
$S^1$,  for any $o\in\cI$ there exist $N_o\in\mathbb{N}$ such that 
$o\in\cI_N$ for any $N\geq N_o$. In fact, by density,  there exists $N_o$ such that $x_{N_o}\in S^1\setminus cl(o)$; hence  $o\in\cI_{N_o}$. 
\end{proof}
So what remains to be shown is that the nets $(\cA,\jmath)_{\cI_N}$ are injective for any $N$. 
We first introduce some suited finite approximations of $S^1$.

\paragraph{Finite approximations of $S^1$: the poset $P_N$.} 
Starting from the sequence $\{x_n\}$ introduced in the previous step, we construct for any $N\in\mathbb{N}$ a finite poset 
$P_N$ associated to the first $N$ elements $x_1,\ldots, x_N$ of the sequence. \smallskip 

The definition of the poset $P_N$ is notably simplified if we assume that the points $x_1 , \ldots, x_N$ 
are ordered as 
\[
 x_i<x_{i+1} \  , \qquad i=1,\ldots N-1 \ ,  
\]
under the clockwise orientation of $S^1$. This does not affect the generality of the proof 
of the injectivity of the net $(\cA,\jmath)_{\cI_N}$ since 
$\cI_N$ does not depend on the order of the points $x_1,\ldots, x_N$. \smallskip

Given the ordered $N$-ple $x_1,\ldots, x_N$,   the elements of the poset $P_N$ 
are the open intervals $(x_i,x_k)$, 
for $i,k=1,\ldots,N$, having, 
with respect to the clockwise orientation, $x_i$ as initial extreme and $x_k$ as 
final extreme respectively, ordered under  inclusion. 
In this way, each $(x_i,x_i)$ is the interval $S^1\setminus\{x_i\}$ and hence
is maximal; on the other hand, the intervals $(x_i,x_{i+1})$, $i \neq N$,
and $(x_N,x_1)$, have contiguous extreme points and hence are minimal.\\
\indent To handle the periodicity with respect to clockwise rotations
we use classes mod $N$, in the following way: for each $n \in \bN$ we denote
its class mod $N$ by $(n)_N \in 1 , \ldots , N$ (here we use the convention
of \S \ref{Db}), and write consequently $x_{(n)_N} \in S^1$. 
Moreover, we introduce the \emph{length function} assigning to each 
$(x_i,x_k) \in P_N$ the positive integer
\[
\ell_{i,k} := \sharp \{ \ j\in\{1,\ldots,N\} : (x_{(j)_N},x_{(j+1)_N}) \subseteq (x_i,x_k) \ \} \ ,
\]
where $\sharp$ stands for the cardinality. In this way each minimal element of $P_N$
has length $1$ and each maximal element has length $N$. Note that the length of
an interval can be easily calculated from the indices,
\[
\ell_{i,k}=(k-i)_N \ . 
\] 
Any interval $(x_i,x_k)$ of length $\ell_{i,k} < N$ has 
only two majorants among the intervals of length $\ell_{i,k}+1$, namely 
$(x_i, x_{(k+1)_N})$ and $(x_{(i-1)_N}, x_{k})$.
Finally, note that with our conventions the points of 
$\{ x_1 , \ldots , x_N \}$ belonging  to $(x_i,x_k)$ are
$x_{(i+1)_N}, x_{(i+2)_N} , \ldots , x_{(i+\ell_{ik}-1)_N}$.
\begin{lemma}
\label{Ea:6a}
The poset 
$P_N$ is isomorphic to the cylinder $C_N$.
\end{lemma}

\begin{proof}
We define a mapping $\rf:P_N\to C_N$ as follows: 
\[
\rf(x_i,x_k):= (i,\ell_{ik}) \ ,\qquad i,k\in\{1,2,\ldots N\} \ . 
\]
This map has inverse
\[
\rf'(i,\ell) := (x_i,x_{(i+\ell)_N})
\ \ , \ \
(i,\ell) \in C_N
\ ,
\]
and is thus bijective. To prove that $\rf$ is order preserving it suffices
to compute, for $\ell_{i,k}< N$
\[
\rf(x_i, x_{(k+1)_N}) = 
(i,\ell_{i,k}+1) >  
(i,\ell_{i,k}) = 
\rf(x_i, x_{(\ell_{i,k}+i)_N}) = 
\rf(x_i, x_{k})
\ ,
\]
and 
\[
\rf(x_{(i-1)_N}, x_{k}) = 
((i-1)_N,\ell_{i,k}+1) > 
(i,\ell_{i,k}) =  
\rf(x_i, x_{(\ell_{i,k}+i)_N}) =
\rf(x_i, x_k)
\ .
\]
So $\rf$ is an isomorphism and the proof follows. 
\end{proof}
\paragraph{Inducing nets over $P_N$.} We now show that the net $(\cA,\jmath)_\cI$ 
induces a net over $P_N$. In the next paragraph we shall see that the restrictions
$(\cA,\jmath)_{\cI_N}$ embed, faithfully on the fibres, into such a nets.\smallskip   

For any $i,k=1,2,\ldots, N$, let  
\begin{equation}
\label{Ea:3}
\cI^{(i,k)}_N:=\big\{ o\in\cI \ | \ cl(o) \subset (x_i,x_k) \big\} \ . 
\end{equation}
According to this definition 
$(x_i,x_k)\notin\cI^{(i,k)}_N$ and $\cI^{(i,k)}_N\subset \cI_N$. 
The poset $\cI^{(i,k)}_N$ is upward directed with respect to inclusion; 
so $\{\cA_o , \jmath_{ao}\}_{\cI^{(i,k)}_N}$ is an inductive system,
and  we define the inductive limit $\rC^*$-algebra
\begin{equation}
\label{Ea:3a}
\widehat \cA_{(i,k)}:= \underrightarrow{\lim} \{\cA_o , \jmath_{ao}\}_{\cI^{(i,k)}_N} \ . 
\end{equation}
The algebras $\widehat \cA_{(i,k)}$ are associated with the interval $(x_i,x_k)$ of $P_N$. 
So we have defined the fibres of a net over $P_N$. We now define the inclusion maps.  
Let $J^{(i,k)}_o:\cA_o\to\widehat\cA_{(i,k)}$ be the canonical embedding for $o\in \cI^{(i,k)}_N$. 
This is a unital monomorphism satisfying the relations
\begin{equation}
\label{Ea:4}
J^{(i,k)}_o\circ \jmath_{oa} =  J^{(i,k)}_a \ , \qquad a\leq o. 
\end{equation}
Note that if $(x_i,x_k)\subseteq (x_j,x_s)$ then $\cI^{(i,k)}_N\subseteq \cI^{(j,s)}_N$. 
We then define, for $a\in\cI^{(i,k)}_N$, 
\begin{equation}
\label{Ea:5} 
\widehat{\jmath}_{(j,s)(i,k)}(J^{(i,k)}_a(A)):= J^{(j,s)}_o \circ \jmath_{oa}(A) 
\ ,  \ \ A\in\cA_a \ , 
\end{equation}
where we take some $o\in\cI^{(j,s)}_N$ with $a\leq o$ since  $\cI^{(i,k)}_N$ is directed. 
We easily find, applying (\ref{Ea:4}), that our definition
does not depend on the choice of $o \geq a$. 
It turns out that $\widehat{\jmath}_{(j,s)(i,k)}$ extends to a unital monomorphism 
from $\widehat\cA_{(i,k)}$ into $\widehat\cA_{(j,s)}$; applying (\ref{Ea:5}), we immediately
find that $\widehat{\jmath}_{(j,s)(i,k)}$ fulfills the net relations
\[
\widehat{\jmath}_{(j,s) (i,k)}\circ \widehat \jmath_{(i,k)(m,r)} 
= 
\widehat\jmath_{(j,s)(m,r)} 
\ , \qquad 
(x_m,x_r)\subseteq (x_i,x_k) \subseteq (x_j,x_s)
\ ,
\]
therefore the system $(\widehat{\cA},\widehat\jmath)_{P_N}$ is a net of $\rC^*$-algebras.
\begin{lemma}
\label{Ea:6}
The net $(\widehat{\cA},\widehat\jmath)_{P_N}$ is injective for any $N$. 
\end{lemma}
\begin{proof}
This follows by Proposition \ref{Db:8} and Lemma \ref{Ea:6a}. 
\end{proof}
\paragraph{Conclusion: the embedding of the nets $(\cA,\jmath)_{\cI_N}$}  We now are ready to prove that any net 
over $S^1$ is injective. 
\begin{theorem}
\label{Ea:11}
Any net of $\rC^*$-algebras $(\cA,\jmath)_{\cI}$ over $S^1$ is injective. 
In particular, the  net $(\cA,\jmath)_{\cI_N}$ embeds, faithfully on the fibres, into $(\widehat{\cA},\widehat\jmath)_{P_N}$ for any  $N$. 
\end{theorem}
\begin{proof}
Injectivity of $(\cA,\jmath)_{\cI}$ follows from Lemma \ref{Ea:2a} and 
Theorem \ref{Beb:10}, once we have shown that $(\cA,\jmath)_{\cI_N}$ is injective for any $N$. 
To this end, as observed in Remark \ref{Bd:12}, it will be  enough to show that $(\cA,\jmath)_{\cI_N}$ embeds, faithfully on the fibres, into 
$(\widehat{\cA},\widehat\jmath)_{P_N}$ because the latter is an injective net.\\   
\indent As a first step we prove that $P_N$ is a quotient of $\cI_N$.
Define  
\begin{equation}
\label{Ea:8}
 \rf(o):= (x_i,x_k) \ , \qquad o\in\cI_N  
\end{equation}
if
\begin{equation}
\label{Ea:9}
\left\{
\begin{array}{ll}
cl(o) \subset (x_i, x_{k}) \ {\mathrm{and}}
\\
x_{(i+1)_N},x_{(i+2)_N}, \ldots, x_{(i+\ell_{ik}-1)_N}  \in cl(o) \ , 
\end{array}
\right.
\end{equation}
where the points listed in the above equation are those  $x_1,x_2,\ldots, x_N$ 
that belong to $(x_i,x_k)$. 
It is clear that $\rf$ is order preserving and surjective, i.e. it is an epimorphism.
As a second step, we define a unital morphism faithful on the fibres  
from $(\cA,\jmath)_{\cI_N}$ into the  
injective net $(\widehat{\cA},\widehat\jmath)_{P_N}$, and this will suffice to
conclude the proof (by functoriality). Given $o\in\cI_N$, define 
\begin{equation}
\label{Ea:10}
\eta_o(A):= J^{\rf(o)}_o(A) \ , \qquad A\in\cA_o \ . 
\end{equation}
So $\eta_o:\cA_o\to\widehat\cA_{\rf(o)}$ is a unital monomorphism. 
Given $a\leq o$, by (\ref{Ea:10}) and (\ref{Ea:5}),  we have 
\[
\eta_o\circ \jmath_{oa} = J^{\rf(o)}_o\circ \jmath_{oa} =
\widehat\jmath_{\rf(o) \rf(a)}\circ J^{\rf(a)}_o \ ,
\] 
and can  conclude that $(\eta,\rf):(\cA,\jmath)_{\cI_N}\to(\widehat\cA,\widehat\jmath)_{P_N}$ 
is a unital morphism faithful on the fibres, completing the proof.
\end{proof}
Using this theorem we deduce the following 
result for $\mathrm{Diff}(S^1)$-covariant nets.
\begin{corollary}
\label{Ea:2}
Any $\mathrm{Diff}(S^1)$-covariant net of
$\rC^*$-algebras over $S^1$ has $H$-invariant states and strongly continuous $H$-covariant 
representations for any amenable subgroup $H$ of $\mathrm{Diff}(S^1)$. 
\end{corollary}
\begin{proof}
Since any net of $\rC^*$-algebras over $S^1$ is injective and 
the homotopy group of $S^1$ is amenable, the proof follows by 
Prop.\ref{Ca:6} and Prop.\ref{Cb:13}.
\end{proof}

We stress that, by Prop.\ref{Cb:13}, the covariant representations
of the previous corollary induce \emph{strongly continuous} $H$-representations 
in the sense of (\ref{Cb:12a}).


\section{Comments and outlook}
\label{Z}

We list some topics and questions arising from the present
paper. \smallskip

\noindent 1. A problem which deserves a further investigation is whether  any net over 
            $S^1$ admits faithful Hilbert space representations. This is a stronger 
            condition than injectivity and is related with 
            the existence of a proper ideal of the fibres of the enveloping net bundle \cite{RVL}.   
              
\noindent 2. An interesting question is whether it is possible to construct 
              representations of nets over $S^1$ directly from the representations of the       
              associated nets over $N$-cylinders \S \ref{Ea}. 
              This could be interesting because $N$-cylinders have 
              some symmetries that might be inherited by the induced 
              representations of the nets over $S^1$.\\[5pt]  
\noindent 3.  Having shown that any net over $S^1$ is injective, it would be  interesting  
              to understand the r\^ole of  representations, 
              carrying a nontrivial representation of the fundamental group of $S^1$,              
              in chiral conformal quantum field theories and, in particular,  
              to relate them to the  new superselection sectors introduced in  
              \cite{CKL}, since they carry the same topological content.\\[5pt] 
\noindent 4. It is hoped that the method used in the present paper to prove injectivity of 
             nets over $S^1$ can  be used for other manifolds. This point is the object of
             a work in progress.\\[5pt]


\appendix

\section{Inductive  limits}
\label{Be}

The constructions that can be made in the category 
of $\rC^*$-algebras can be also made in the category 
of nets of $\rC^*$-algebras. The direct sum of two nets, for instance,  
%
%
has fibres and inclusion maps given by, respectively, 
the direct sum of the fibres and of the inclusion maps 
of the two nets. On the same line the tensor product is defined. 
In this appendix, however, we focus on a single construction 
which will turns out to be very useful for analysing injectivity: 
the inductive limit. We show that the category of nets of $\rC^*$-algebras
has inductive limits. The functor assigning 
the enveloping net bundle turns out to preserve inductive limits, and
this implies that inductive limits preserve injectivity.


\subsection{Basic properties} 
\label{Bea}
A \emph{inductive system} of nets of $\rC^*$-algebras is given by the following data: 
an upward directed poset $\Lambda$ (we shall denote the elements of $\Lambda$ by Greek letters $\alpha, \beta$, 
etc..., and the order relation by $\preceq$); 
a family of nets of $\rC^*$-algebras 
$(\cA^\alpha,\jmath^\alpha)_{K^\alpha}$, with $\alpha\in \Lambda$,
and  a family  of unital morphisms
$(\psi^{\alpha\si}, \rf^{\alpha\si}): 
 (\cA^\si,\jmath^\si)_{K^\si}\to (\cA^\alpha,\jmath^\alpha)_{K^\alpha}$ 
for $\si\preceq \alpha$, with $\rf^{\alpha\si}$ injective, 
such that 
\begin{equation}
\label{Bea:0}
(\psi^{\alpha\si}, \rf^{\alpha\si})\circ (\psi^{\si\delta}, \rf^{\si\delta}) = 
(\psi^{\alpha\delta}, \rf^{\alpha\delta})\ ,  \qquad \delta\preceq \si\preceq \alpha \ .
\end{equation}
We shall denote such an inductive system by 
\begin{equation}
\label{Bea:0a}
\big\{ (\cA^\alpha,\jmath^\alpha)_{K^\alpha}, \, (\psi^{\si\alpha}, \rf^{\si\alpha}) \big\}_\Lambda \ , 
\end{equation}
and call $(\psi^{\si\alpha}, \rf^{\si\alpha})$ the \emph{linking morphisms}  of the system. \bigskip

We stress that \emph{do not assume} in the definition of inductive system 
that the linking morphisms are monomorphisms, but  
we only require that $\rf^{\alpha\si}$ is a monomorphism of posets.  \smallskip 

Note that by the definition of morphisms of nets, for any $\alpha\preceq \si$,  
\begin{equation}
\label{Bea:0b}
\psi^{\si\alpha}_a\circ \jmath^\alpha_{ae}=   
\jmath^\si_{\rf^{\si\alpha}(a) \rf^{\si\alpha}(e)}\circ \psi^{\si\alpha}_e \ , \qquad e\leq^\alpha a  \ ,
\end{equation}
where $\leq^\alpha$ is the order relation of $K^\alpha$.  
Moreover, it is worth observing that an inductive system of posets $\{ K^\alpha \ , \ \rf^{\si\alpha}\}_\Lambda$ 
is associated with our inductive system of nets. In the following we show that any 
inductive system of nets of $\rC^*$-algebras has an inductive limit, which turns out to be 
a net of $\rC^*$-algebras over the inductive limit poset. \bigskip

First of all we explain the inductive limit poset of $\{K^\alpha , \rf^{\alpha\si}\}_\Lambda$. 
This is a poset $K$, with order relation denoted by $\leq$,  
such that for any $\alpha\in \Lambda$ there is a  monomorphism 
$F^\alpha: K^\alpha\to K$ satisfying the following properties: 
\begin{enumerate}
\item[(1)]  $\rF^\alpha\circ \rf^{\alpha\beta}=\rF^{\beta}$ for any $\beta\preceq \alpha$ ; 
\item[(2)]  $K=\cup_{\alpha\in \Lambda} \rF^\alpha(K^\alpha)$ ; 
\item[(3)]  for any other poset $K'$  with a family of morphisms $\rH_\alpha:K^\alpha\to K'$
            such that $\rH^\alpha\circ \rf^{\alpha\beta}=\rH^{\beta}$ there exists a unique morphism 
            $\rH:K\to K'$ such that $\rH\circ \rF^\alpha = \rH^\alpha$ for any $\alpha\in \Lambda$. 
\end{enumerate}
Existence of the inductive limit poset can be proved as a consequence of the more general construction 
of colimits in the setting of small categories (see \cite[\S III.3]{ML}).\\
\indent Now, given $o\in K$, we consider
\[
\Lambda_o 
\ := \
\{ \alpha \in \Lambda \ : \ o \in \rF^\alpha(K^\alpha) \}
\ \subseteq \Lambda
\ .
\]
Clearly $\Lambda_o$ is not empty because of property (2) of the inductive limit poset. 
\begin{lemma}
\label{Bea:0c}
If $K$ is the inductive limit poset of $(K^\alpha,\rf^{\si\alpha})_\Lambda$,
then the following properties hold:
(i) $\Lambda_o$ is an upper set
\footnote{An upper set $P$ of $\Lambda$ is a subset such that if $\si\in \Lambda$ and 
there is $\alpha\in P$ such that $\si\preceq \alpha$, then $\si\in P$.} 
of $\Lambda$ for any $o\in K$ and, as a subposet of $\Lambda$, it is upward directed.  
(ii) For any  $o\in K$ there is a map $\Lambda_o\ni \alpha \mapsto o_\alpha\in K^\alpha$
such that
\begin{equation}
\label{Bea:0d}
\begin{array}{ll}
 \rf^{\si\alpha}(o_\alpha) = o_\si \ , & \alpha\preceq \si \ ; \\
 \rF^\alpha(o_\alpha) = o \ , &   \alpha\in \Lambda_o \ . 
\end{array}
\end{equation}
\end{lemma}
\begin{proof}
$(i)$ Given $\alpha\in \Lambda_o$ the element $a\in K^\alpha$ satisfying 
$\rF^\alpha(a)=o$ is unique because  $\rF^\alpha$ is injective,
so we denote this element by $o_\alpha$. Now, given $\alpha\in \Lambda_o$,  
and  $\beta\in S$ with $\alpha\preceq \beta$ then $\beta\in \Lambda_o$, since 
$\rF^{\beta}(\rf^{\beta\alpha}(o_\alpha)) = \rF^{\alpha}(o_\alpha) = o$, 
so  $\Lambda_o$ is an upper set of $\Lambda$. Thus $\Lambda_o$ is upward directed because
$\Lambda$ is. $(ii)$ has been proved in the course of proving  $(i)$.
\end{proof}

We now construct the inductive limit of 
$\big\{ (\cA^\alpha,\jmath^\alpha)_{K^\alpha}, \, 
 (\psi^{\si\alpha}, \rf^{\si\alpha}) \big\}_\Lambda$. 
Given $o\in K$, using the properties of the function $\Lambda_o\ni\alpha\mapsto o^\alpha$
and applying (\ref{Bea:0d}), (\ref{Bea:0b}), we see that 
$\{\cA^{\alpha}_{o_\alpha},\psi^{\si\alpha}_{o_{\alpha}}\}_{\Lambda_o}$ 
is an inductive system of $\rC^*$-algebras over $\Lambda_o$. 
So we can define the $\rC^*$-inductive limit
\begin{equation}
\label{Bea:1}
\cA_o:=\underrightarrow{\lim} \cA^{\alpha}_{o_\alpha}
\ .
\end{equation}
Now, it is worth recalling 
some facts about inductive limits (see \cite{Bla1}). First of all,  
there is a unital morphism 
$\psi^{\alpha}_o:\cA^\alpha_{o_\alpha}  \to \cA_o$  
such that 
\begin{equation}
\label{Bea:2}
\psi^{\si}_o \circ \psi^{\si\alpha}_{o_{\alpha}} = \psi^{\alpha}_o  \ , \qquad \alpha\preceq \si \ .  
\end{equation}
Furthermore, the unital ${^*}$-algebra ${\cA}'_o$ defined by  
\begin{equation}
\label{Bea:2a}
\cA'_o:= \bigcup_{\alpha\in \Lambda_o} \psi^{\alpha}_o(\cA^\alpha_{o_{\alpha}})
\end{equation}
is a dense ${^*}$-subalgebra of $\cA_o$ for any $o\in K$. 
Finally, the norm $\|\cdot \|_{\infty}$ of the inductive limit
satisfies, for any $A\in\cA^\alpha_a$,  the relation
\begin{equation}
\label{Bea:3}
\| \psi^{\alpha}_o(A)\|_{\infty} = \inf_{\si\succeq \alpha} \|\psi^{\si\alpha}_{o_\alpha}(A)\| = 
\lim_{\si \succeq \alpha} \|\psi^{\si\alpha}_{o_{\alpha}}(A)\| \ ,   
\end{equation}
because  the norm  $\|\psi^{\si\alpha}_{o_\alpha}(A)\|$ is monotone decreasing 
in $\si$. \smallskip  
 
The correspondence $\cA: K\ni o\to \cA_o$, where $\cA_o$ is the unital $\rC^*$-algebra 
defined by equation (\ref{Bea:1}),  is  the fibre  of the inductive limit net over $K$. 
What is yet missing are 
the inclusion  maps. Given $o,\tilde o \in K$ with $o \leq \tilde o $, we first define 
the inclusion map $\jmath_{\tilde oo}$ on  the ${^*}$-algebra $\cA'_o$. 
Afterwards we shall prove that these maps can be isometrically extended to  all of 
$\cA_o$. Given $o, \tilde o$ as above, take $\alpha\in \Lambda_o$ and choose 
$\si\in \Lambda_{\tilde o}$ such that  $\si\succeq \alpha$, and, for any $A\in\cA^\alpha_a$,  define  
\begin{equation} 
\label{Bea:4}
\jmath_{\tilde o o}(\psi^{\alpha}_o(A)) : = 
\psi^{\si}_{\tilde o} \circ \jmath^{\si}_{{\tilde o}_\si o_{\sigma}} \circ 
\psi^{\si\alpha}_{o_\alpha} (A)  \ . 
\end{equation}
This definition does not depend on the choice 
of  $\si\in \Lambda_{\tilde o}$ with $\si\succeq \alpha$. In fact, take $\gamma\in \Lambda_{\tilde o}$ 
with $\gamma\succeq \si$. 
Relations (\ref{Bea:2}), (\ref{Bea:0b}),(\ref{Bea:0d}), and (\ref{Bea:0}) give
\begin{align*}
\psi^{\si}_{\tilde o} \circ \jmath^{\si}_{{\tilde o}_\sigma o_\sigma}\circ 
\psi^{\si\alpha}_{o_\alpha} & =   \psi^{\gamma}_{\tilde o}\circ \psi^{\gamma\sigma}_{{\tilde o}_\si} \circ 
\jmath^{\sigma}_{{\tilde o}_\si o_\si}\circ \psi^{\sigma\alpha}_{o_\alpha} \\
& = 
\psi^{\gamma}_{\tilde o}\circ \jmath^{\gamma}_{\rf^{\gamma\sigma}({\tilde o}_\sigma) \rf^{\gamma\sigma}(o_\sigma)}\circ 
\psi^{\gamma\sigma}_{o_\sigma} \circ 
\psi^{\sigma\alpha}_{o_\alpha} \\ 
& = 
\psi^{\gamma}_{\tilde o}\circ \jmath^{\gamma}_{{\tilde o}_\gamma o_\gamma}\circ 
\psi^{\gamma\alpha}_{o_\alpha}  \ .  
\end{align*}
By  the definition of $\cA'_o$ 
$\jmath_{\tilde o o}$ maps from $\cA'_o$ into 
$\cA'_{\tilde o}$ since $\Lambda_o$ is upward directed.  
We now prove that these maps satisfy all the properties 
of inclusion maps. 
\begin{lemma}
\label{Bea:5}
Given an inductive system 
$\big\{ (\cA^\alpha,\jmath^\alpha)_{K^\alpha}, \, (\psi^{\si\alpha}, \rf^{\si\alpha}) \big\}_\Lambda$
of nets of $\rC^*$-algebras 
the following assertions hold:
\begin{itemize}
\item[(i)] The triple $(\cA,\jmath)_K$ is a net of $\rC^*$-algebras. 
\item[(ii)] If the system is composed of $\rC^*$-net bundles, then $(\cA,\jmath)_K$ is a $\rC^*$-net bundle.
\end{itemize}
\end{lemma}
\begin{proof}
$(i)$ By definition $\jmath_{\tilde oo}: \cA'_{\tilde o}\to \cA'_{o}$ 
is a unital morphism such that 
\[
\jmath_{\hat o \tilde o}\circ \jmath_{\tilde oo} = \jmath_{\hat oo}  \ , \qquad o\leq \tilde o \leq \hat o \ . 
\]
We now prove that $\jmath_{\tilde oo}$ are isometries. To this end, given   
$\alpha\in \Lambda_o$, 
take $\si\in \Lambda_{\tilde o}$ with $\si\succeq \alpha$. By (\ref{Bea:3}), (\ref{Bea:0b}) and (\ref{Bea:0})
we have 
\begin{align*}
\| \jmath_{\tilde o o} \circ \psi^{\alpha}_o(A) \|_{\infty} & = 
\inf_{\gamma\succeq\si} 
\| \psi^{\gamma\si}_{{\tilde o}_\si} \circ
   \jmath^{\si}_{{\tilde o}_\si o_\si} \circ 
   \psi^{\si\alpha}_{o_\alpha} (A) \|  = 
\inf_{\gamma\succeq \si} 
\| \jmath^{\gamma}_{{\tilde o}_\gamma o_\gamma} \circ 
   \psi^{\gamma\sigma}_{o_\sigma} \circ
   \psi^{\sigma\alpha}_{o_\alpha} (A) \| \\ & = 
\inf_{\gamma\succeq \si} 
\| \jmath^{\gamma}_{{\tilde o}_\gamma o_\gamma} \circ 
   \psi^{\gamma\alpha}_{o_\alpha}(A) \|  = 
\inf_{\gamma\succeq \si} 
\| \psi^{\gamma\alpha}_{o_\alpha}(A) \| \\ & = 
\|\psi^{\alpha}_o(A) \|_{\infty} \ , 
\end{align*}
since the $\jmath^{\gamma}_{{\tilde o}_\gamma o_\gamma}$ are isometries.  
So $\jmath_{\tilde oo}$ admits an isometric extension to  a mapping
from $\cA_o$ into $\cA_{\tilde o}$. \smallskip 

\noindent
$(ii)$ We study the image of $\jmath_{\tilde oo}$. First of all we show that the following relations
hold for every $\alpha \in \Lambda_o$,
\begin{equation}
\label{Bea:5a}
\jmath_{\tilde o o} \circ \psi^{\alpha}_o(A) = 
\psi^{\alpha}_{\tilde o} \circ \jmath^{\alpha}_{{\tilde o}_\alpha o_{\alpha}}(A) 
\ , \qquad 
A \in \cA^\alpha_{o_\alpha} \ ;
\end{equation} 
in fact, given $\si\in \Lambda_{\tilde o}$ with $\si\succeq \alpha$, by \ref{Bea:4} and (\ref{Bea:2}) we have
\begin{align*}
\jmath_{\tilde o o} \circ \psi^{\alpha}_o(A)  & = 
\psi^{\si}_{\tilde o} \circ \jmath^{\si}_{{\tilde o}_\si o_{\sigma}} \circ
\psi^{\si\alpha}_{o_\alpha} (A) \\ 
&  = 
\psi^{\si}_{\tilde o} \circ 
\psi^{\si\alpha}_{\tilde o_\alpha} \circ 
\jmath^{\alpha}_{{\tilde o}_\alpha\, o_{\alpha}}(A) = 
\psi^{\alpha}_{\tilde o} \circ \jmath^{\alpha}_{{\tilde o}_\alpha\, o_{\alpha}}(A) \ . 
\end{align*}
Since every
$\jmath^{\alpha}_{{\tilde o}_\alpha o_{\alpha}} : 
 \cA^\alpha_{o_\alpha} \to 
 \cA^\alpha_{\tilde o_\alpha}$
is an isomorphism, the previous relation gives  
$\jmath_{\tilde o o} \circ \psi^{\alpha}_o (\cA^\alpha_{o_\alpha}) = 
\psi^{\alpha}_{\tilde o} (\cA^\alpha_{\tilde o_\alpha})$.
So, $\jmath_{\tilde oo}(\cA'_o) = \cA'_{\tilde o}$, and this, in turns, implies that 
$\jmath_{\tilde oo}$ extends to an isomorphism from $\cA_o$ to $\cA_{\tilde o}$, completing the proof.  
\end{proof}

Finally, we have the following result. 
\begin{proposition}
\label{Bea:6}
Let $\big\{(\cA^\alpha,\jmath^\alpha)_{K^\alpha}, 
\, (\psi^{\si\alpha}, \rf^{\si\alpha})\big\}_\Lambda$ be  an inductive system of nets of $\rC^*$-algebras. 
Then for each $\alpha\in \Lambda$ there is a unital morphism
\[
(\Psi^{\alpha}, \rF^\alpha) : (\cA^\alpha,\jmath^\alpha)_{K^\alpha} \to (\cA,\jmath)_K
\]
such that 
\begin{itemize}
\item[(i)] $(\Psi^{\si}, \rF^\si)\circ (\psi^{\si\alpha}, \rf^{\si\alpha}) = 
            (\Psi^{\alpha}, \rF^\alpha)$ for any $\alpha\preceq \si$ ;
\item[(ii)] The image $\Psi^{\alpha}_{o_\alpha} (\cA^\alpha_{o_\alpha})$, as $\alpha$ varies in $\Lambda_o$,
is dense in  $\cA_o$ for any $o\in K$; 
\item[(iii)] If there is a net of $\rC^*$-algebras $(\cB,y)_K$ and a collection of morphisms
 $(\Phi^\alpha,\rF^\alpha) : (\cA^\alpha,\jmath^\alpha)_{K^\alpha} \to (\cB,y)_K$, $\alpha\in \Lambda$,  
 such that 
 $(\Phi^\si,\rF^\si)\circ (\psi^{\si\alpha}, \rf^{\si\alpha}) = (\Phi^\alpha,\rF^\alpha)$, 
 $\alpha \preceq \sigma$,
 then there exists a unique morphism $\Phi : (\cA,\jmath)_K \to (\cB,y)_K$ such that 
 $\Phi \circ (\Psi^\alpha,\rF^\alpha) = (\Phi^\alpha,\rF^\alpha)$ for any $\alpha\in \Lambda$. 
\end{itemize}
\end{proposition}
\begin{proof}
Define 
\begin{equation}
\label{Bea:6a}
\Psi^{\alpha}_a := \psi^\alpha_{\rF^\alpha(a)} \ , \qquad a\in K^\alpha  \ . 
\end{equation}
By relation (\ref{Bea:5a}) we have  
\[
\Psi^\alpha_a\circ \jmath^\alpha_{ae} = 
\psi^\alpha_{\rF^\alpha(a)}\circ \jmath^\alpha_{ae} = 
 \jmath_{\rF^\alpha(a)\rF^\alpha(e)}\circ \psi^\alpha_{\rF^\alpha(e)} = 
 \jmath_{\rF^\alpha(a)\rF^\alpha(e)}\circ \Psi^\alpha_{e} \ , 
\]  
and this proves that 
$(\Psi^{\alpha}, \rF^\alpha):(\cA^\alpha,\jmath^\alpha)_{K^\alpha}\to (\cA,\jmath)_K$  
is a unital morphism. \smallskip

\noindent 
In this way $(i)$ and $(ii)$ are easy to prove. 
$(iii)$ follows from the universal property  of inductive limits of $\rC^*$-algebras: 
in fact, if $\alpha\in \Lambda_o$ then
$\Phi^\alpha_{o_\alpha}:\cA^\alpha_{o_\alpha}\to \cB_o$ 
is a collection of  morphisms satisfying the relation 
\[
\Phi^\si_{o_\si}\circ \psi^{\si\alpha}_{o_\alpha} = \Phi^\alpha_{o_\alpha}\ ,
\]
by the universal property of the inductive limit $\cA_o$ there is a unique morphism 
$\Phi_o: \cA_o\to\cB_o$ such that 
$\Phi_o \circ  \Psi^\alpha_{o_\alpha} = \Phi^\alpha_{o_\alpha}$. So let  $\Phi$ denotes the collection 
$\Phi_o$, with $o\in K$. We now prove that $\Phi$ is a morphism 
of nets.  Take $o,\tilde o\in K$ and $\alpha\in \Lambda_o$, $\si\in \Lambda_{\tilde o}$  as in  
Definition (\ref{Bea:5a}). Then we have 
\begin{align*}
\Phi_{\tilde o}\circ \jmath_{\tilde o o}\circ \Psi^\alpha_{o_\alpha}  & 
= 
\Phi_o\circ \Psi^\si_{{\tilde o}_\si} \circ \jmath^{\si}_{{\tilde o}_\si o_\sigma} \circ 
\psi^{\si\alpha}_{o_\alpha} 
= \Phi_o\circ \Psi^\si_{{\tilde o}_\si} \circ 
\psi^{\si\alpha}_{{\tilde o}_\alpha} \circ \jmath^{\alpha}_{{\tilde o}_\alpha o_\alpha }  \\
& = 
\Phi^\si_{{\tilde o}_\si} \circ 
\psi^{\si\alpha}_{{\tilde o}_\alpha} \circ \jmath^{\alpha}_{{\tilde o}_\alpha o_\alpha }  
= \Phi^\alpha_{{\tilde o}_\alpha}  \circ \jmath^{\alpha}_{{\tilde o}_\alpha o_\alpha }  
= y_{{\tilde o} o} \circ \Phi^\alpha_{{o}_\alpha}  \\
& =  y_{{\tilde o} o} \circ \Phi_o \circ \Psi^\alpha_{o_\alpha}  
\end{align*}
for any $\alpha$. Density implies that $\Phi:(\cA,\jmath)_K\to (\cB,y)_K$ 
is a morphism, and it is clear that 
$\Phi\circ (\Psi^\alpha,\rF^\alpha) = (\Phi^\alpha,\rF^\alpha)$.
Finally, uniqueness follows in a similar fashion. 
\end{proof}
Given a net of $\rC^*$-algebras satisfying
the properties of the previous proposition, 
an inductive system $\big\{(\cA^\alpha,\jmath^\alpha)_{K^\alpha}, 
\, (\psi^{\si\alpha}, \rf^{\si\alpha})\big\}_\Lambda$, will be called the 
\emph{inductive limit} of the system and, from now on, will be denoted by 
$\underrightarrow{\lim}(\cA^\alpha,\jmath^\alpha)_{K^\alpha}$.  
The property $(iii)$ of the proposition is the \emph{universal} property 
of  inductive limits.


\subsection{Injectivity of inductive limits} 
\label{Beb}

By functoriality, to any inductive system of nets there corresponds an inductive system 
of  enveloping net bundles. 
The functor of taking  the enveloping net bundle commutes with  inductive limits. 
As a consequence, injectivity is preserved under inductive limits indicating   a strategy 
for analyzing  injectivity. \bigskip

Consider an inductive system of nets of $\rC^*$-algebras 
$\big\{ (\cA^\alpha,\jmath^\alpha)_{K^\alpha} , \, 
(\psi^{\si\alpha},\rf^{\si\alpha})\big\}_\Lambda$. 
Let $(\overline{\cA}^\alpha,\overline{\jmath}^\alpha)_{K^\alpha}$ be the enveloping 
net bundle of $(\cA^\alpha,\jmath^\alpha)_{K^\alpha}$, and let 
\[
\e^\alpha : 
(\cA^\alpha,\jmath^\alpha)_{K^\alpha} \to 
(\overline{\cA}^\alpha,\overline{\jmath}^\alpha)_{K^\alpha}
\]
be the canonical embedding. By the universal property  
of the enveloping net bundle, for any $\alpha\preceq \si$ there is a unique morphism
$(\overline{\psi}^{\si\alpha},\rf^{\si\alpha}) : 
 (\overline{\cA}^\alpha,\overline{\jmath}^\alpha)_{K^\alpha}\to 
 (\overline{\cA}^\si,\overline{\jmath}^\si)_{K^\si}$
satisfying the relation.  
\begin{equation}
\label{Beb:1}
(\overline{\psi}^{\si\alpha},\rf^{\si\alpha})\circ \e^{\alpha} = 
\e^\si\circ (\psi^{\si\alpha},\rf^{\si\alpha}) \ . 
\end{equation}
This implies 
\begin{equation}
\label{Beb:2}
(\overline{\psi}^{\si\alpha},\rf^{\si\alpha})\circ (\overline{\psi}^{\alpha\beta},\rf^{\alpha\beta}) = 
(\overline{\psi}^{\si\beta},\rf^{\si\beta}) \ , \qquad \beta\preceq \alpha\preceq\si 
\ ,
\end{equation}
hence  $\{(\overline{\cA}^\alpha,\overline{\jmath}^\alpha)_{K^\alpha}, \, (\overline{\psi}^{\si\alpha},\rf^{\si\alpha})\}_\Lambda$ is an inductive system of $\rC^*$-net bundles; 
\emph{the system of the enveloping net bundles}. \smallskip 
 
Note that even if the linking morphisms of the original system 
are monomorphism, in general the linking morphisms of the system of  
enveloping net bundles may not be  monomorphisms. This depends on the injectivity 
of the nets of the original system.\smallskip

Let $\underrightarrow{\lim} (\overline{\cA}^\alpha,\overline{\jmath}^\alpha)_{K^\alpha}$ be the $\rC^*$-net bundle inductive limit of the system of the enveloping net bundles, and denote 
by \smash{$(\overline{\Psi}^\alpha, \rF^\alpha)$} be the embedding of the nets 
$(\overline{\cA}^\alpha,\overline{\jmath}^\alpha)_{K^\alpha}$ into this limit. 
We now show that this limit is nothing but 
the enveloping net bundle of the net 
$\underrightarrow{\lim} (\cA^\alpha,\jmath^\alpha)_{K^\alpha}$. 

\begin{proposition}
\label{Beb:3}
Let 
$\{(\cA^\alpha,\jmath^\alpha)_{K^\alpha},(\psi^{\si\alpha},\rf^{\si\alpha})\}_\Lambda$
be an inductive system. Then 
the inductive limit $\underrightarrow{\lim} (\overline{\cA}^\alpha,\overline{\jmath}^\alpha)_{K^\alpha}$ 
is isomorphic  to the enveloping 
net bundle of the inductive limit $\underrightarrow{\lim} (\cA^\alpha,\jmath^\alpha)_{K^\alpha}$. 
\end{proposition}
\begin{proof}
Let $(\overline{\cA},\overline{\jmath})_K$ be the enveloping net bundle 
of $\underrightarrow{\lim} (\cA^\alpha,\jmath^\alpha)_{K^\alpha}$ 
and 
$\e: \underrightarrow{\lim} (\cA^\alpha,\jmath^\alpha)_{K^\alpha}\to (\overline{\cA},\overline{\jmath})_K$ 
denote the canonical embedding.
We start by showing that $(\overline\cA,\overline\jmath)_{K}$ embeds into 
$\underrightarrow{\lim} (\overline{\cA}^\alpha,\overline{\jmath}^\alpha)_{K^\alpha}$. 
To this end, we use the universal property of  inductive limits of nets (Proposition \ref{Bea:6}.$iii$):
given $\alpha\in \Lambda$, note that 
$(\overline{\Psi}^\alpha,\rF^\alpha)\circ \e^\alpha : 
 (\cA^\alpha,\jmath^\alpha)_{K^\alpha} \to 
 \underrightarrow{\lim}(\overline{\cA}^{\alpha},\overline{\jmath}^\alpha)_{K^\alpha}$ 
is a morphism such that 
\begin{equation}
\label{Beb:4}
(\overline{\Psi}^\alpha,\rF^\alpha)\circ \e^\alpha\circ (\psi^{\alpha\beta},\rf^{\alpha\beta}) = 
(\overline{\Psi}^\alpha,\rF^\alpha)\circ (\overline{\psi}^{\alpha\beta},\rf^{\alpha\beta}) \circ \e^\beta =
(\overline{\Psi}^\beta,\rF^\beta)\circ \e^\beta
\ ,
\end{equation}
so, by the universal the property  of  the inductive limit  
$\underrightarrow{\lim} (\cA^\alpha,\jmath^\alpha)_{K^\alpha}$,  
there exists a morphism
$\theta: \underrightarrow{\lim} (\cA^\alpha,\jmath^\alpha)_{K^\alpha} \to 
\underrightarrow{\lim}(\overline{\cA}^{\alpha},\overline{\jmath}^\alpha)_{K^\alpha}$ 
such that 
\begin{equation}
\label{Beb:5}
\theta \circ (\Psi^\alpha,\rF^\alpha) = (\overline{\Psi}^\alpha,\rF^\alpha)\circ \e^\alpha 
\ , \qquad \alpha\in \Lambda 
\ .
\end{equation}
On the other hand, since $\underrightarrow{\lim}(\overline{\cA}^{\alpha},\overline{\jmath}^\alpha)_{K^\alpha}$ 
is a $\rC^*$-net bundle, the universal property of the enveloping net bundle 
$(\overline{\cA},\overline{\jmath})_{K}$ 
says that there is a unique morphism 
$\Theta : 
 (\overline{\cA},\overline{\jmath})_{K} \to 
 \underrightarrow{\lim}(\overline{\cA}^{\alpha},\overline{\jmath}^\alpha)_{K^\alpha}$
such that 
\begin{equation}
\label{Beb:6}
\Theta\circ \e = \theta. 
\end{equation}
We now prove that there is a morphism in the opposite direction. Consider the morphisms
\[
 \e\circ (\Psi^\alpha,\rF^\alpha) : 
 (\cA^\alpha,\jmath^\alpha)_{K^\alpha} \to 
 (\overline{\cA},\overline{\jmath})_{K}
 \ \ , \ \
 \alpha \in S
 \ .
\]
The universal property  of the enveloping net bundle $(\overline{\cA}^\alpha,\overline{\jmath}^\alpha)_{K^\alpha}$
implies that there are morphisms
$(\chi^\alpha,\rF^\alpha) :
 (\overline{\cA}^\alpha,\overline{\jmath}^\alpha)_{K^\alpha} \to 
 (\overline{\cA},\overline{\jmath})_{K}$ ,
intertwining the canonical embeddings    
\begin{equation}
\label{Beb:7}
(\chi^\alpha,\rF^\alpha)\circ \e^\alpha= \e\circ (\Psi^\alpha,\rF^\alpha)
\ ,
\end{equation}
and compatible with the inductive structures of the enveloping net bundles
\begin{equation}
\label{Beb:8}
(\chi^\alpha,\rF^\alpha)\circ (\overline{\psi}^{\alpha\si},\rf^{\alpha\si}) = 
(\chi^\si,\rF^\si) 
\ .
\end{equation}
By the universal property of  inductive the limit     
$\underrightarrow{\lim}(\overline{\cA}^{\alpha},\overline{\jmath}^\alpha)_{K^\alpha}$, 
there is a unique morphism 
$\Theta' : 
 \underrightarrow{\lim}(\overline{\cA}^{\alpha},\overline{\jmath}^\alpha)_{K^\alpha} \to 
 (\overline{\cA},\overline{\jmath})_{K}$ 
satisfying the equation
\begin{equation}
\label{Beb:9}
 \Theta' \circ (\overline{\Psi}^\alpha,\rF^\alpha) = (\chi^\alpha,\rF^\alpha) \ , \qquad \alpha\in \Lambda \ .
\end{equation}
We now prove that $\Theta$ is the inverse of $\Theta'$. 
First, using equations 
(\ref{Beb:6}), (\ref{Beb:5})  (\ref{Beb:9}) and (\ref{Beb:7}),
we note that for any $\alpha\in \Lambda$, 
\begin{align*}
\Theta'\circ \Theta\circ \e\circ (\Psi^\alpha,\rF^\alpha) & =
\Theta'\circ \theta\circ (\Psi^\alpha,\rF^\alpha) \\
&  =\Theta'\circ (\overline{\Psi}^\alpha,\rF^\alpha)\circ \e^\alpha
 = (\chi^\alpha,\rF^\alpha)\circ \e^\alpha \\
 &  = \e\circ (\Psi^\alpha,\rF^\alpha) \ ,
\end{align*}
so $\Theta'\circ \Theta\circ \e= \e$. By (\ref{Bd:11})  
we conclude that $\Theta'\circ \Theta$ is the identity automorphism of $(\overline\cA,\overline{\jmath})_K$.  

Conversely,  equations (\ref{Beb:9}), (\ref{Beb:7}), (\ref{Beb:6})
and (\ref{Beb:5}), for any $\alpha\in \Lambda$, imply  
\begin{align*}
\Theta\circ \Theta' \circ (\overline{\Psi}^\alpha,\rF^\alpha)\circ \e^\alpha 
&  = 
\Theta\circ (\chi^\alpha,\rF^\alpha)\circ \e^\alpha \\
&  = \Theta\circ \e\circ (\Psi^\alpha,\rF^\alpha)
 = \theta\circ (\Psi^\alpha,\rF^\alpha)\\
& = (\overline{\Psi}^\alpha,\rF^\alpha)\circ \e^\alpha \ , 
\end{align*}
thus (\ref{Bd:11}), implies that
$\Theta\circ \Theta' \circ (\overline{\Psi}^\alpha,\rF^\alpha) = (\overline{\Psi}^\alpha,\rF^\alpha)$ 
for any $\alpha$, and $\Theta\circ \Theta' $ is the identity automorphism of $\underrightarrow{\lim}(\overline{\cA}^{\alpha},\overline{\jmath}^\alpha)_{K^\alpha}$.
\end{proof}
We now prove the main result of the present section.
\begin{theorem}
\label{Beb:10}
Let $\{(\cA^\alpha,\jmath^\alpha)_{K^\alpha},\,(\psi^{\si\alpha},\rf^{\si\alpha})\}_\Lambda$ 
be an inductive system of net of $\rC^*$-algebras. If the linking morphisms 
are monomorphisms and the nets 
$(\cA^\alpha,\jmath^\alpha)_{K^\alpha}$ 
are all injective, then  the inductive limit $\underrightarrow{\lim} (\cA^\alpha,\jmath^\alpha)_{K^\alpha}$ is an injective net.  
\end{theorem}
\begin{proof}
It is enough to prove that the morphism 
\[
\theta: 
\underrightarrow{\lim} (\cA^\alpha,\jmath^\alpha)_{K^\alpha} \to 
\underrightarrow{\lim}(\overline{\cA}^{\alpha},\overline{\jmath}^\alpha)_{K^\alpha}
\ ,
\]
defined in the previous proposition, is a monomorphism.  
To this end, considering $o\in K$ and $A\in\cA^\alpha_{o_\alpha}$, and applying (\ref{Beb:5}),
Proposition \ref{Bea:6}.$i$ and (\ref{Beb:1}), we find
\begin{align*}
\| \theta_o \circ \Psi^\alpha_{o_{\alpha}}(A) \|_\infty & =
\|\overline{\Psi}^\alpha_{o} \circ \e^\alpha_{o_\alpha}(A) \|_{\infty}
 =\inf_{\si\geq \alpha}\|\overline{\psi}^{\si\alpha}_{o_\alpha} \circ \e^\alpha_{o_\alpha}(A) \| \\
&  =\inf_{\si\geq \alpha}\|\e^\si_{o_\si} \circ \psi^{\si\alpha}_{o_\alpha} (A) \| \\
&  =\|A \|  \ , 
\end{align*}
as both $\e^\si_{o_\si}$ and $\psi^{\si\alpha}_{o_\alpha}$ are isometries.
Since $\cup_\alpha \Psi^\alpha_{o_{\alpha}}(\cA^\alpha_{o_\alpha})$
is dense in $\cA_o$ (Proposition \ref{Bea:6}.$ii$), we conclude that $\theta_o$ is an isometry 
for any $o\in K$; hence $\theta$ is a monomorphism.
\end{proof}

\subsection{On the continuity condition}
\label{AppB}

We prove that the notion of a continuous symmetry group $G$ of a poset $K$ as given in the present 
paper implies that introduced in \cite{RVL}. We shall use the simplicial set 
associated to the poset and the corresponding notion of homotopy equivalence of paths. 
For all these notions and related results we refer the reader to the cited paper.

\begin{lemma}
Let $G$ be a continuous symmetry group of a poset $K$. Then 
for any path $p:o\to a$ there  $\tilde o,\tilde a\in K$, with 
$o\leq \tilde o$ and $a\leq \tilde a$,   
and an open neighbourhood $U$ of the identity $e$ of $G$ such that 
$ga\leq \tilde a$ and $go\leq o$ and 
\[
 (\tilde a a) *p*\overline{(\tilde oo)} \sim  (\tilde a g(a)) *gp*\overline{(\tilde o g(o))} \ , \qquad g\in U  \ . 
\]
\end{lemma}
\begin{proof}
We give a proof by induction. Let $b$ a 1-simplex.  By continuity of the action of $G$ there is $O\in K$ such that $|b|\ll O$. Note in particular that the faces 
of the 1-simplex satisfy $\partial_0b,\partial_1b\ll O$. Let $V$ 
be  the neighbourhood of identity of $G$ associated to $|b|\ll O$. Then 
\begin{equation}
\label{AppB:1}
 (O,\partial_0b)*b*\overline{(O,\partial_1b)} \sim (O,g(\partial_0b))*g(b)*\overline{(O,g(\partial_1b))} \ , \qquad g\in V \ .
\end{equation}
In fact, note that all the elements of the poset involved in the above relation are 
smaller than $O$ for any $g\in V$. Then  homotopy equivalence follows  
because any upward directed poset is simply connected.\\
\indent Assume that the above relation holds for paths which are composition of $n$ 1-simplices. 
Let $p:o\to a$ be such a path and let $b$ a 1-simplex such that $\partial_0b=o$. 
By hypothesis there are $o\leq \tilde o$ and $a\leq \tilde a$,   
and an open neighbourhood $W$ of the identity $e$ of $G$ such that 
$ga\leq \tilde a$ and $go\leq o$ and 
\begin{equation}
\label{AppB:2}
 (\tilde a a) *p*\overline{(\tilde oo)} \sim  (\tilde a g(a)) *g(p)*\overline{(\tilde o g(o))} \ , \qquad g\in W \ .
\end{equation}
Let $O$ and $V$ be as in the equation (\ref{AppB:1}).  Since $\tilde o,O \gg o$, there exists 
$o'$ such that $o\ll o' \ll \tilde o, O$. Let $V'$ be the neighbourhood of the indentity 
of $G$ associated to $o\ll o'$. If $U:= V\cap W\cap V'$, then the equation 
(\ref{AppB:1}) and (\ref{AppB:2}) are verified for any $g\in U$. Furthermore since 
$o\leq o'\leq \tilde o$, we have $(\tilde o,o')*(o'o)\sim (\tilde o,o)$. Since 
homotopy equivalence is stable under composition, we have   
$\overline{(\tilde o o)} * (\tilde o o')\sim \overline{(o' o)}$. 
This and equation (\ref{AppB:2}) yield
\begin{equation}
\label{AppB:3}
 (\tilde a a) *p*\overline{(o'o)} \sim  (\tilde a g(a)) *g(p)*\overline{(o' g(o))} \ , \qquad g\in U \ . 
\end{equation}
The same argument applied to equation (\ref{AppB:1}) yields
\begin{equation}
\label{AppB:4}
 (o',o)*b*\overline{(O,\partial_1b)} \sim (o',g(o))*g(b)*\overline{(O,g(\partial_1b))} \ , \qquad g\in U \ , 
\end{equation}
(recall that $o=\partial_0b$). The composition of the left hand sides of the equations 
(\ref{AppB:3}) (\ref{AppB:4}) gives 
\begin{equation}
\label{AppB:5}
(\tilde a a) *p*\overline{(o'o)} * (o',o)*b*\overline{(O,\partial_1b)} \sim 
(\tilde a a) *p*b*\overline{(O,\partial_1b)} \ , 
\end{equation}
while the composition of the right hand sides gives 
\begin{equation}
\label{AppB:6}
(\tilde a g(a)) *g(p)*\overline{(o' g(o))} * (o',g(o))*g(b)*\overline{(O,g(\partial_1b))} 
\sim (\tilde a g(a)) *g(p*b)*\overline{(O,g(\partial_1b))}
\end{equation}
for any $g\in U$. Finally,
the equations (\ref{AppB:3}),(\ref{AppB:4}),(\ref{AppB:5}) and (\ref{AppB:6}) give 
\[
(\tilde a a) *p*b*\overline{(O,\partial_1b)}\sim (\tilde a g(a)) *g(p*b)*\overline{(O,g(\partial_1b))} \ , \qquad g\in U \ , 
\]
completing the proof.
\end{proof}

\noindent {\small {\bf Aknowledgements.} 
We gratefully acknowledge the hospitality
and support  of  the Graduate School of 
Mathematical Sciences of the University of Tokyo, where part of this paper has been developed, 
in particular Yasuyuki Kawahigashi for his warm hospitality. We also would like to thank 
\emph{all} the operator algebra group of the University of Roma ``Tor Vergata'', 
Sebastiano Carpi and Fabio Ciolli,
for the several fruitful discussions on the topics treated in  this paper.}


\end{document}